\documentclass[10pt]{article}

\usepackage{amsmath}
\usepackage{graphicx}
\usepackage{natbib}
\usepackage{url}
\usepackage{enumerate}
\usepackage{amsfonts,amssymb,mathrsfs,dsfont,amsthm}
\usepackage{color}
\usepackage{bbm,bm}
\usepackage{comment}
\usepackage{bigints}
\usepackage[normalem]{ulem}
\usepackage[utf8]{inputenc}
\usepackage{float}
\usepackage{multirow}
\usepackage{array}
\RequirePackage[colorlinks,citecolor=blue,urlcolor=blue]{hyperref}
\usepackage{chngcntr}
\usepackage{algorithm}

\newtheorem{theorem}{Theorem}%[chapter]
\newtheorem{definition}{Definition}%[section]
\newtheorem{lemma}{Lemma}[section]
\newtheorem{proposition}{Proposition}%[section]
\theoremstyle{definition}
\newtheorem{remark}{Remark}[section]

\def\l{\left}
\def\r{\right}

\newcolumntype{C}[1]{>{\centering\arraybackslash}p{#1}}

\newcommand{\ddr}{\mathrm{d}}

\newcommand{\fine}{\hfill $\Box$}
\newcommand{\comillas}[1]{``\,#1\,''}
\definecolor{iblue}{rgb}{0.1,0,0.75}
\definecolor{ired}{rgb}{0.9,0,0.1}
\definecolor{ipurple}{rgb}{0.54,0.17,0.89}
\definecolor{gray}{rgb}{0.5, 0.5, 0.5}

\newcommand{\Ind}{\mathds{1}}

\newcommand{\tp}{\tilde{p}}
\newcommand{\tx}{\tilde{x}}

\newcommand{\Prob}{\mathbb{P}}
\newcommand{\X}{\mathbb{X}}
\newcommand{\Y}{\mathbb{Y}}
\newcommand{\R}{\mathbb{R}}

\newcommand{\D}{\mathcal{D}}

\newcommand{\iid}{\overset{iid}{\sim}}
\newcommand{\ind}{\overset{ind}{\sim}}
\newcommand{\Be}{\mathsf{Be}}

\newcommand{\clJ}{\mathcal{J}}
\newcommand{\clA}{\mathcal{A}}

\newcommand{\rest}{\cdots}

\newcommand{\pier}{{i,j,\bm\gamma}}

\newcommand\restr[2]{{% we make the whole thing an ordinary symbol
  \left.\kern-\nulldelimiterspace % automatically resize the bar with \right
  #1 % the function
  \littletaller % pretend it's a little taller at normal size
  \right|_{#2} % this is the delimiter
  }}
\newcommand{\littletaller}{\mathchoice{\vphantom{\big|}}{}{}{}}

\allowdisplaybreaks

\title{\bf Speeding up the ordered allocation sampler}

\author{
  Mar\'ia F. Gil--Leyva\\
  Department of Probability and Statistics, IIMAS--UNAM, M\'exico\\
  \texttt{marifer@sigma.iimas.unam.mx} \\
  \and
  Fidel Selva\\
  Faculty of Science--UNAM, M\'exico\\
  \texttt{fidel@ciencias.unam.mx} \\
  \and
  Pierpaolo De Blasi\\
  Collegio Carlo Alberto and ESOMAS Department,\\ University of Torino, Italy\\
  \texttt{pierpaolo.deblasi@unito.it} \\
}

\date{}

\begin{document}
\maketitle

\begin{abstract}
The ordered allocation sampler is a Gibbs sampler designed to explore the posterior distribution in nonparametric mixture models. It encompasses both infinite mixtures and finite mixtures with random number of components, and it has be shown to possess mixing properties that pair well with collapsed, or marginal, samplers that integrate out the mixing distribution. The main advantage is that it adapts to mixing priors that do not enjoy tractable predictive structures needed for the implementation of marginal sampling methods. Thus it is as widely applicable as other conditional samplers while enjoying better algorithmic performances. In this paper we provide a modification of the ordered allocation sampler that enhances its performances in a substantial way while easing its implementation. In addition, exploiting the similarity with marginal samplers, we are able to adapt to the new version of the sampler the split-merge moves of Jain and Neal. Simulation studies confirm these findings.
%Posterior inference in Bayesian nonparametric mixture models customarily relies on Gibbs sampling methods \red{that} can been classified into marginal or conditional samplers, depending on whether they integrate out or include the mixing distribution. 
%While marginal samplers are known to enjoy better mixing properties than their conditional counterparts, the latter can be used to implement a much wider range of \red{prior distributions} \sout{models}. The ordered allocation sampler is a conditional sampler that mixes on the same space \red{\it[xx is it exactly the same? xx]} as marginal samplers do.
%Thus, it remains 
 \end{abstract}

\noindent%
\vfill

\newpage
%--------------------------------------------------------
%--------------------------------------------------------

\section{Introduction}

Mixture models are well suited to model exchangeable heterogeneous data where observations can be clustered into groups, each group coming from a different population uniquely identified by a mixture component. 
Each component brings its own parameter, which determines the distribution of the corresponding population, and its own weight that represents the probability that an observation is associated to such population. The Bayesian approach consists in taking component parameters and weights as random, that is placing a prior distribution on the mixing distribution.
Since the number of components or distinct populations is rarely known in advance, it is convenient to consider a random and arbitrarily large number of components \citep[e.g.][]{Ric:Gre:97,Mil:Har:18} or an infinite number of them \citep[cf.][]{Lo:84}. However, posterior inference for these types of mixtures models can be challenging due to the need to explore parameter spaces of random or infinite dimensions. Gibbs sampling methods are commonly employed to sample from the posterior distribution via the introduction of allocation variables that identify each observation with one of the mixture components. They can be broadly divided into marginal and conditional samplers. The first marginal algorithm was proposed by \citet{Esc:94,Esc:Wes:95} in the case of Dirichlet process mixing prior, see \citet{Nea:00} for successive developments.
%and numerous variants have been designed since then \citep[cf.][]{Esc:Wes:95} \red{MacEachern(1995)} \red{MacEachern(1998)} \citep{Nea:00}. 
Marginal samplers are termed this way because they integrate out the weights and unobserved component parameters. By doing so, one can avoid the problem of dealing with a varying and potentially infinite model dimension and instead design a sampler that evolves in the space of partitions defined by the ties in allocation variables. This marginalization naturally yields good mixing properties, cf.  Algorithm 2 and 8 of \citet{Nea:00}. However, it requires a certain level of mathematical tractability of the mixing prior, specifically a generalized P\'olya urn scheme representation of the predictive distributions, \citep{Bla:Mac:73,Pit:96ims} as characterized by the  \emph{exchangeable partition probability function} (EPPF). 
Unfortunately, this representation is only available for a handful of mixing priors, preventing marginal samplers from being widely applicable. Conditional samplers such as the truncated Gibbs sampler by \cite{Ish:Jam:01}, the slice samplers by \cite{Walk:07} and \cite{Kall:etal:11}, and the retrospective sampler by \cite{Pap:Rob:08} stand as a feasible alternative in such cases. These samplers condition on all the weights and component parameters, thus requiring some sort of truncation or augmentation to deal with the infinite model dimension. A major drawback in comparison with marginal samplers is that they evolve in the space of cluster's labels, instead of evolving in the space of partitions, which hinders their mixing \citep{Port:etal:06}. 

The ordered allocation sampler by \cite{Deb:Gil:23} is a conditional sampler that address these issues by conditioning on component parameters and weights in the order in which they are discovered by the data. The sampler operates with allocation variables that preserve the least element order of the blocks of the partition. Since $n$ data points can not be generated from more than $n$ distinct components, the sampler at most requires to update $n$ component parameters and weights, and in practice much less. In particular, there is no need to recur to any truncation or model augmentation to deal with a potentially infinite number of mixture components. A key advantage of working with {\it ordered} allocation variables is that the sampler evolves in the space of ordered partitions, hence in a space that is similar to that where marginal samplers do. This endows the ordered allocation sampler with better mixing properties than other conditional samplers whilst retaining their wide applicability. The restricted support of the allocation variables, however, demands taking some caution in its implementation. In compliance with the least element order, the first variables have less freedom in allocating the corresponding observations to a mixture component when compared with the last ones. Thus, the initial ordering of data points can affect the mixing properties of the  sampler. \cite{Deb:Gil:23} suggested to randomly permute the data at each Gibbs scan by exploiting the exchangeability of the data. This acceleration step has been found to be crucial to assure better mixing properties when compared with other conditional samplers. There remains however other issues that need care. One disadvantage is the limited ability of the sampler to vary the number of discovered components in a single scan.
%, in fact worse than in the case of marginal sampler. 
Moreover, the updating of the allocation variables can be tricky due to the need to identify first the \emph{set of admissible moves}, that is moves that preserve the least element order. See Section 2 for more details.
%at at each iteration many components can be discovered, but only one observed component can become unobserved. This is not the case for other samplers in which any number of components can be discovered or emptied at any iteration. Moreover, each time an ordered allocation variable is updated, it is necessary to compute a \emph{set of admissible moves} as determined by the appropriate full conditional. This might cause some troubles in programming and slows down running. 

We effectively tackle these issues by designing a variant of the ordered allocation sampler that evolves in the space of unordered partitions, alike marginal samplers, instead of the space of ordered partitions with blocks in the least element order.
% \red{\it[xx is it? xx]} 
As we will prove, this change of space is possible due to the exchangeability of observations together with distributional symmetries of weights labelled in order of appearance \citep[e.g.][]{Pit:96}. 
%As it turns out the law of this arrangement of weights remains invariant under size-biased permutations . 
The new version of the ordered allocation sampler not only surpasses the original one in every aspect, but it makes it possible to adapt tools that enhance the performance of marginal methods. In particular, as a second contribution we derive \emph{split-merge moves} in the vein of  \cite{Jai:Nea:04,Jai:Nea:07} but tailored to the ordered allocation sampler.
As in the original version by Jain and Neal, these moves can be used as an acceleration step for the Gibbs sampler or they can be used as a standalone Monte Carlo Markov Chain (MCMC) sampling scheme. 
%While these can be used on their own to define a valid MCMC method, in the simulation studies we will add them an extra step to the Gibbs sampler to speed up convergence to the objective distribution.

Before we delve into details, it is important to note that, while the ordered allocation sampler was originally designed for both infinite mixtures and mixtures with random number of components, in this paper we focus only on infinite mixtures. However, the improvements we propose inhere can be readily applied to mixtures with a random number of components. 
%we propose inhere apply identically for these three distinct scenarios, but for simplicity here we focus on infinite mixture models.
Infinite mixtures models are those where the mixing distribution is modeled via the Dirichlet process \citep{Fer:73,Esc:Wes:95}, the Pitman-Yor process \citep{Pit:Yor:92,Pit:Yor:97}, normalized random measures with independent increments \citep{Reg:etal:03,Fav:etal:16}, and exchangeable stick-breaking \citep{Gil:Men:21}, among others.

\section{Ordered allocation sampler}\label{sec:OAS}

In this section we recall the theoretical setting of the ordered allocation sampler (OAS) for posterior inference on  Bayesian nonparametric mixture models. Let the data be denoted by $\bm{y}_{1:n} = (y_i)_{i=1}^n$ with values in a Borel space $\mathbb{Y}$. The data points are taken as exchangeable observations, specifically conditionally independent and identically distributed (iid) from a random mixture of densities
\begin{equation}\label{eq:mix}
  \int_{\X} g(y\mid x) P(dx) = \sum_{j=1}^\infty p_j g(y\mid x_j).
\end{equation}  
where $g(\cdot\mid x)$ is a density over $\Y$ for each $x$  in a Borel parameter space $\X$. The mixing distribution $P = \sum_{j=1}^\infty p_j \delta_{x_j}$ is a proper \emph{species sampling model} \citep[see][]{Pit:96ims,Pit:06} over $\X$.
%(with Borel $\sigma$-field $\B_{\X}$). 
This means that the atoms $\bm{x} = (x_j)_{j=1}^{\infty}$ of $P$, or mixture component parameters, are iid from a diffuse distribution, $\nu$, called base measure, and are independent of the weights, $\bm{p} = (p_j)_{j=1}^{\infty}$, which are non-negative random variables that sum up to one. With the introduction of an exchangeable \emph{species sampling sequence}, $\bm{\theta} = (\theta_i)_{i=1}^{\infty}$, driven by $P$, we have
\begin{equation}\label{eq:sss}
y_i \mid \theta_i \ind g(\cdot\mid \theta_i), \quad \theta_i \mid P \iid P, \quad i \geq 1.
\end{equation}
In applications one can envision two different inferential goals, estimation of the density of observation $y_i$'s, and clustering of the data through the allocation of each $y_i$ to a single mixture component. Posterior inference can be carried out by MCMC methods or variational methods \citep[cf.][]{Blei:Jor:06}. Among MCMC methods, Gibbs sampler algorithms are the most commonly employed. They can be broadly classified as conditional or marginal samplers. Conditional samplers \citep[e.g.][]{Ish:Jam:01,Walk:07,Pap:Rob:08,Kall:etal:11} are characterized by the inclusion of the mixing distribution, $P$, in the sampler by updating its atoms, $\bm{x}$, and weights, $\bm{p}$. In contrast, marginal methods \citep[e.g.][]{Nea:00} integrate out $P$ and model directly the law of $\bm{\theta}$ through 
%$\theta_1 \sim \nu$ together with 
the prediction rule  $\Prob(\,\theta_{i+1}\in \cdot\mid \bm{\theta}_{1:i}\,)$ of the species sampling sequence (here $\bm{\theta}_{1:i} = (\theta_1,\ldots,\theta_i)$). While marginal samplers enjoy better mixing properties than conditional ones \citep[see][]{Port:etal:06}, they cannot be used when the prediction rule, $\Prob(\,\theta_{i+1}\in \cdot\mid \bm{\theta}_{1:i}\,)$ is not available in explicit form, which is the case for most species sampling models.
%Strictly speaking,
The ordered allocation sampler (OAS) \citep{Deb:Gil:23} is a conditional sampler as it includes the mixing distribution, however, operationally it is more similar to a marginal sampler because it relies on the conditional prediction rule of the species sampling sequence given the atoms and weights of $P$ in order of appearance. More precisely, it is motivated by the following result.

\begin{theorem}\label{theo:rep_sss}
Let $\bm{\theta} = (\theta_i)_{i=1}^\infty$ be a species sampling sequence driven by the species sampling process $P = \sum_{j=1}^\infty p_j \delta_{x_j}$, with base measure $\nu$. Define the $j$th distinct value to appear in $\bm{\theta}$ through $\tx_j = \theta_{M_j}$, where
$M_{j} = \min\{i > M_{j-1} :
  \theta_i \not\in \{\tx_1,\ldots,\tx_{j-1}\}\}$,
for $j \geq 2$ and $M_1 = 1$. Then,
\begin{enumerate}[i.]
\item
$\tx_j = x_{\alpha_j}$, for every $j \geq 1$, and some $\bm{\alpha} = (\alpha_j)_{j=1}^{\infty}$ satisfying
\begin{equation}\label{eq:size-biased_pick}
\begin{split}
  \Prob(\alpha_1 = j\mid \bm{p})
  &= p_j, \\
  %\quad j \in [m],\\
  \Prob(\alpha_{l+1} = j\mid \bm{p},\bm{\alpha}_{1:l})
  &=\frac{p_j}{1-\sum_{i=1}^{l}p_{\alpha_i}}
  \Ind_{\{j \not\in \bm{\alpha}_{1:l}\}},
  \quad l \geq 1.
\end{split}
\end{equation}
with $\bm{\alpha}_{1:l} = (\alpha_1,\ldots,\alpha_l)$. In particular, $\bm{\tx} = (\tx_j)_{j=1}^{\infty}$ are iid from the base measure $\nu$.
\item
The almost sure limits
  $$
  \tilde{p}_j
  = \lim_{n \to \infty}\frac{|\{i \leq n: \theta_i
  = \tilde{x}_j\}|}{n}, \quad j \geq 1,
  $$
exist, $\tp_j > 0$, $\sum_{j=1}^\infty \tp_j = 1$ almost surely, and $\tp_j = p_{\alpha_j}$, for every $j \geq 1$, with $\bm{\alpha}$ as in i.
% In particular, $\bm{\tp} = (\tp_j)_{j=1}^{\infty}$ is invariant under size-biased permutations.
\item
$\theta_1 = \tx_1$, and the conditional prediction rule of $\bm{\theta}$ given $\bm{\tp}$ and $\bm{\tx}$ is
  $$
  \Prob(\theta_{i+1}\in \cdot
  \mid  \bm{\tp},\bm{\tx},\bm{\theta}_{1:i})
  = \sum_{j=1}^{k_i}\tilde{p}_j\delta_{\tilde{x}_j}
  + \bigg(1- \sum_{j=1}^{k_i}\tilde{p}_j\bigg)
  \delta_{\tilde{x}_{k_i +1}},
  $$
for every $i \geq 1$, where $k_i$ is the number of distinct values in $\bm{\theta}_{1:i} = (\theta_1,\ldots,\theta_i)$.
\item
$\bm{\tp}$, $\bm{p}$ and $\bm{\alpha}$ are independent of $\bm{\tx}$.
\end{enumerate}
%\end{enumerate}
\end{theorem}

Notice that $\bm{\alpha}$ and $\bm{\tp} = (\tp_j)_{j=1}^{\infty}$ as in Theorem \ref{theo:rep_sss} are obtained by sampling without replacement from $\{1,2,\ldots\}$ and $\{p_1,p_2,\ldots\}$ respectively, with probabilities $(p_1,p_2,\ldots)$. Thus, the long run frequencies $\bm{\tp}$ of the species sampling sequence $\bm{\theta}$ in order of appearance are {\it in size-biased order}, or in other words, they are invariant under size-biased permutation \citep[ISBP,][]{Pit:96}.
%It is easy to see that if $\bm{\tilde{\alpha}} = (\tilde{\alpha}_l)_{l=1}^{\infty}$ satisfies $\Prob[\tilde{\alpha}_1 = j\mid \bm{\tp}] = \tp_j$ and $\Prob[\tilde{\alpha}_{l+1} = j\mid \bm{\tp}] \propto \tp_j\Ind_{\{j \not\in \bm{\tilde{\alpha}}_{1:l}\}}$, for $l \geq 1$, analogously to \eqref{eq:size-biased_pick}, then $\bm{\tp}$ is equal in distribution to $(\tp_{\tilde{\alpha}_j})_{j=1}^{\infty}$. \red{\it[xx isn't it exactly what we say in ii? xx]}
%If so, it is said that $\bm{\tp}$ is invariant under size-biased permutation (ISBP) or that it is in size-biased random order. \red{\it[xx I do not like to use this acronim xx]} 
The converse of Theorem \ref{theo:rep_sss} is also true in the sense that if the distinct values, $\bm{\tx}$, that $\bm{\theta}$ exhibits in order of appearance are iid from $\nu$, the sequence of long-run proportions $\bm{\tp}$ is ISBP, $\bm{\tp}$ is independent of $\bm{\tx}$, and \emph{iii} in Theorem \ref{theo:rep_sss} holds, then $\bm{\theta}$ is a species sampling sequence driven by $P = \sum_{j=1}^{\infty}\tp_j\delta_{\tx_j} = \sum_{j=1}^{\infty}p_{\alpha_j}\delta_{x_{\alpha_j}} = \sum_{j=1}^{\infty}p_{j}\delta_{x_{j}}$. Theorem \ref{theo:rep_sss} follows from results derived by \cite{Pit:95,Pit:96}, and the proof of this exact statement can be found in the supplementary material of \cite{Deb:Gil:23}.

Taking inspiration from Theorem \ref{theo:rep_sss}, the OAS works with the component parameters and the weights in order of appearance, $\bm{\tx}$ and $\bm{\tp}$, as well as the \emph{ordered allocation variables}, $\bm{d}_{1:n} = (d_i)_{i=1}^n$ given by $d_i = j$ if and only if $\theta_i = \tx_j$. This way $y_i\mid (d_i,\bm{\tx}) \ind g(\cdot \mid \tx_{d_i})$, $i \geq 1$ and it follows from \emph{iii} in Theorem \ref{theo:rep_sss} that $d_1 = 1$ and
\[
\Prob(d_{i+1} \in \cdot \mid \bm{d}_{1:i},\bm{\tp}) = \sum_{j=1}^{k_i}\tilde{p}_j\delta_{j}
  + \bigg(1- \sum_{j=1}^{k_i}\tilde{p}_j\bigg)
  \delta_{k_i +1},
\]
for every $i \geq 1$ where $k_i = \max\{\bm{d}_{1:i}\}$ coincides with the number of blocks featured by the first $i$ observations. This yields the augmented likelihood
\begin{equation}\label{eq:like_sb}
\pi(\bm{y}_{1:n},\bm{d}_{1:n} \mid \bm{\tp},\bm{\tx})
= \prod_{j=1}^{k_n}\tp_j^{n_j-1}
  \bigg(1-\sum_{l=1}^{j-1}\tp_l\bigg)
  \prod_{i \in D_j}g(y_i\mid\tx_{j})
  \Ind_{\mathcal{D}},
\end{equation}
where $k_n = \max\{\bm{d}_{1:n}\}$, $D_j = \{i \leq n: d_i = j\}$, $n_j = |D_j|$, and $\mathcal{D}$ is the event that $\{D_1,\ldots,D_{k_n}\}$ is a partition of $\{1,\ldots,n\}$ with blocks in the \emph{least element order}, specifically $D_j \neq \emptyset$, for $j \leq k_n$, and $\min\l(D_1\r)< \min\l(D_2\r) < \cdots < \min\big(D_{{k_n}}\big)$. Now, in order to compute the full conditional distributions required at each iteration of the sampler, it is necessary to specify the prior distribution of $\bm{\tx}$ and $\bm{\tp}$, which takes the form
\begin{equation}\label{eq:prior_sb}
\pi(\bm{\tx},\bm{\tp}) = \prod_{j\geq 1}\nu(\tx_j) \times \pi(\bm{\tp}).
\end{equation}
For some species sampling models such as the Dirichlet process and its two-parameter generalization, the Pitman-Yor process, \citep{Fer:73,Seth:94,Pit:Yor:92,Pit:Yor:97,Ish:Jam:01}, it is well known that the weights that are ISBP 
%have prior law $\pi(\bm{\tp})$ that 
admit the stick-breaking representation
\[
\tp_1 = v_1, \quad \tp_j = v_j\prod_{l=1}^{j-1}(1-v_l), \quad j \geq 2,
\]
where $v_j \ind \Be(1-\sigma,\beta+j\sigma)$, $j \geq 1$ for some $\sigma \in [0,1)$ and $\beta > -\sigma$. Thus, it is clear that one can (indirectly) specify the prior law of $\bm{\tp}$ via $\bm{v}=(v_i)_{i=1}^\infty$.
Normalized random measures with independent increments (NRMIs) \citep{Fav:etal:12,Fav:etal:16} stand as another example of species sampling processes for which a similar characterization of the law of weights in size-biased order is available. Unfortunately, for most species sampling models, the law of this arrangement of weights remains intractable. To keep the OAS applicable in these instances, one can model $\bm{\tp}$ through the weights in no particular order, $\bm{p}$, together with a permutation of their indexes, $\bm{\alpha}$ as in \eqref{eq:size-biased_pick}, that indicate the order in which elements of $\bm{p}$ appear, or are discovered, in an exchangeable sequence driven by $P$. Namely we set $\tp_j = p_{\alpha_j}$ and replace \eqref{eq:like_sb} and \eqref{eq:prior_sb} with
\begin{equation}\label{eq:like}
\pi(\bm{y}_{1:n},\bm{d}_{1:n} \mid \bm{p},\bm{\alpha},\bm{\tx})
= \prod_{j=1}^{k_n}p_{\alpha_j}^{n_j-1}
  \bigg(1-\sum_{l=1}^{j-1}p_{\alpha_l}\bigg)
  \prod_{i \in D_j}g(y_i\mid\tx_{j})
  \Ind_{\mathcal{D}},
\end{equation}
and
\begin{equation}\label{eq:prior}
\pi(\bm{\tx},\bm{p},\bm{\alpha}) = \prod_{j\geq 1} \nu(\tx_j) \times \prod_{j\geq 1} p_{\alpha_j}\bigg(1-\sum_{l=1}^{j-1}p_{\alpha_l}\bigg)^{-1}\Ind_{\clA}\times \pi(\bm{p}).
\end{equation}
respectively, where the event $\clA$ indicates $\alpha_j \neq \alpha_l$ for $j \neq l$. Since $\bm{p}$ has no more distributional requirements other than $\sum_{j=1}^{\infty} p_j = 1$ and $p_j \geq 0$, it is much simpler to (possibly indirectly) specify $\pi(\bm{p})$ than it is to specify $\pi(\bm{\tp})$.

The full conditional distributions that are required to update the random elements are proportional to the product of \eqref{eq:like_sb} and \eqref{eq:prior_sb} or to the product of \eqref{eq:like} and \eqref{eq:prior}. Here we recall how to update the ordered allocation variables, which is the part that concerns the modification we will propose. The full conditional distributions of $\bm{\tx}$ and $\bm{\tp}$ (or $\bm{p}$ and $\bm{\alpha}$) are detailed in Appendix \ref{app:OAS} \citep[also see ][ for further details]{Deb:Gil:23}.

%\subsection*{Updating of ordered allocation variables.}

As noted by \cite{Deb:Gil:23}, the ordered allocation variables $\bm{d}_{1:n}$ cannot be updated independently of each other as in other conditional samplers. Instead, similarly to marginal samplers, they need to be updated
%\emph{incrementally}, 
\emph{sequentially}, that is one at a time by conditioning on the current value of the remaining ones. We find
\begin{equation}\label{eq:d_post_prev0}
  \Prob(\,d_i = d\mid \rest\,)
  \propto \tp_d\, g(y_i\mid\tx_d) \times \prod_{j=1}^{k_n}
  \tp_j^{-1}\bigg(1-\sum_{l=1}^{j-1}\tp_l\bigg)\Ind_{\D}.
\end{equation}
To sample from this distribution one first identifies the set, $\D_i$, of {\it admissible moves} for $d_i$, which contains all positive integers $d$ for which $\Ind_{\D}=1$ under the assumption $d_i = d$.
%i.e. they preserve the least element order. 
%To spell this out, let $k_{-i} = \max\{d_l:l \neq i\}$ and $k_{<i} = \max\{d_l: l < i\}$ be the number of blocks after the removal of the $i$th observation. Hence, $k_{-i} =k_n-1$ when $i$ is in a block on its own, $k_{-i} =k_n$ otherwise. We can anticipate that in the first case $\D_i=\{d_i\}$ and no updating takes place, whenever $i<n$ and $k_n>d_i$, otherwise the least element order is violated.
For example, if before updating $d_3$ we observe $(d_1,d_2,d_3,d_4) = (1,1,2,3)$, then $\mathcal{D}_3 = \{2\}$, as for any other choice of $d_3 \neq 2$ the least element order of the resulting blocks is violated. In this example no updating of $d_3$ would take place. 
In general, we can first restrict the set $\D_i$ to $d \in \{1,\ldots,k_{i-1},k_{i-1}+1\}$, later we keep $d \in \D_i$ whenever $d_i=d$ preserves the least element order of the corresponding blocks, $D_j= \{l:d_l = j\}$. 
%That is, for $d \in \N$, define $k_n^{(d)} = \max\{k_{-i},d\}$, where $k_{-i} = \max\{d_l:l \neq i\}$,
%and add $d \in \D_i$ if, under the assumption $d_i = d$, the sets $D_j= \{l\leq n:d_l = j\}$ are non-empty, for $j \leq k_n^{(d)}$, and satisfy $\min\l(D_1\r)< \min\l(D_2\r) < \cdots < \min\big(D_{{k^{(d)}_n}}\big)$.
Now, set $k_{-i} = \max\{d_l: l \neq i\}$ and for each $d \in \D_i$ set $k_n^{(d)} = \max\{k_{-i},d\}$, which corresponds to the number of blocks of the partition obtained by setting $d_i=d$. Since either $k_n^{(d)} = k_{-i}$ ($i$ is added to an existing block) or $k_n^{(d)} = k_{-i}+1$ ($i$ forms a new block), we can divide  \eqref{eq:d_post_prev0} by $\prod_{j=1}^{k_{-i}}\tp_j^{-1}(1-\sum_{l=1}^{j-1}\tp_l)$ obtaining
\begin{equation}\label{eq:d_post_prev}
  \Prob(\,d_i = d\mid \rest\,) \propto
  \begin{cases}
  \tp_d g(y_i\mid \tx_d) & \text{ if }  k_n^{(d)} = k_{-i},\\
\l(1-\sum_{l=1}^{k_{-i}}\tp_l\r)g(y_i\mid \tx_d) & \text{ if } k_n^{(d)} = k_{-i}+1,
\end{cases}
\end{equation}
%Once we have identified ,
for $d \in \D_i$, and $\Prob(\,d_i = d\mid \rest\,) = 0$ otherwise.
%Now, since the support of \eqref{eq:d_post_prev} is precisely $\D_i \subseteq \{1,\ldots,k_{-1}+1\}$
It is straightforward to sample from this distribution once $\D_i$ is found as the latter is a finite set.

As a consequence of the way in which the ordered allocation variables are updated three main issues arise in this algorithm.
(a) The initial ordering of the data points might affect the mixing properties of the sampler. In fact, the set of admissible moves, $\D_i$, of $d_i$ is always contained in $\{1,\ldots,k_{i-1}+1\}$, with $k_0 = 0$ and $k_i = \max\{\bm{d}_{1:i}\}$, in particular $\D_1 = \{1\}$.
%Then for $d_1$ we always have $\D_1 = \{1\}$, for $d_2$ we get $\D_2 = \{1,2\}$ \blue{$= \{1,k_1+1\}$
Recalling that $d_i$ indicates from which component of the mixture was $y_i$ sampled, this means that while the latest data points can be reallocated to any of the occupied components or to a new one, the first data points will rarely be reassigned to a different component.
(b)
%While many components can be created at each iteration, only one can be emptied.
While many components can be created in a single iteration, only the latest group can be emptied because of the restriction that the least element order imposes. This in turns causes slow mixing of the number of occupied components. To explain this, say that before updating $\bm{d}_{1:n}$ we have $k_n$ groups and let $l = \min\{D_{k_n}\}$ be the minimum of the latest block. Since $l \in D_{k_n}$, for each $j \in \{1,\ldots,k_n-1\}$, must be at least one index $i < l$ such that $i \in D_j$, at all times before updating $d_l$.
%when updating any allocation variable before $d_l$, there must be at least one index $i < l$ such that $i \in D_j$, for each $j \in \{1,\ldots,k_n-1\}$. 
Thus by the time we update $d_l$, all the first $k_n-1$ blocks must have at least one recently updated ordered allocation variable $d_i$, with $i \leq l$, assigned to them, which means that they can not be emptied at the current iteration.
%Because of the least element order, the groups that can be emptied are the last ones, but in order to have, say, the last two disappear in a single iteration, the last needs to be emptied first, and bla bla bla. This causes slow mixing of the number of occupied components.
(c) From a computational viewpoint it is a drawback that for each index $i \in \{1,\ldots,n\}$, one has to compute the set of admissible moves $\D_i$ at every iteration.

To overcome the first issue, \cite{Deb:Gil:23}  proposed to
%exploit the exchangeability of the data and 
perform, at each iteration of the sampler, a uniform random permutation, $\rho$, of the data, obtaining $\bm{y}'_{1:n} = (y'_1,\ldots, y'_n) = (y_{\rho(1)},\ldots,y_{\rho(n)})$. This is valid because $\rho$ does not depend on $\bm{y}_{1:n}$, so $\pi(y_1,\ldots,y_n) = \pi(y_{\rho(1)},\ldots,y_{\rho(n)})$ by exchangeability.
%include an acceleration step in which the data set $\bm{y}_{1:n}$ is randomly  permuted obtaining $\bm{y}'_{1:n} = (y'_1,\ldots, y'_n) = (y_{\rho(1)},\ldots,y_{\rho(n)})$. 
In accordance, one relabels the ordered allocation variables, $\bm{d}'_{1:n} = (d'_1,\ldots,d'_n)$, so to preserve the same partition of the data, and that now the least element order of the blocks holds for $D'_j = \{i: d'_i = j\}$. We must rearrange as well the observed component parameters and weights, say $\bm{\tx}'_{1:k_n} = (\tx'_1,\ldots,\tx'_{k_n})$ and $\bm{\tp}'_{1:k_n} = (\tp'_1,\ldots,\tp'_{k_n})$, in the order in which they are discovered by $\bm{y}'_{1:n}$. In particular, if $\bm{\tp}$ is being modelled indirectly through $\bm{p}$ and $\bm{\alpha}$, it is enough to set $\tp'_j = p_{\alpha'_j}$ where $\bm{\alpha}'_{1:k_n} = (\alpha'_1,\ldots,\alpha'_{k_n})$, indicates which weight in $\bm{p}$ was the $j$th one discovered by $\bm{y}'_{1:n}$. While this acceleration step solves the first issue, significantly improving the mixing of the OAS, it does not tackle the other ones.

\section{Efficient ordered allocation sampler}\label{sec:EOAS}

Here we propose a different way of updating the ordered allocation variables that effectively tackles the three issues pointed out in the previous section.
The key idea is to exploit a permutation of the data each time a $d_i$ is updated instead of permuting the data once per iteration.
However, as we will see, in practice it is not necessary to permute the data every time: it is sufficient to update $d_i$ as if $y_i$ was the last observation, in a manner reminiscent of marginal samplers.
%The key theoretical idea is to merge the acceleration step in which data is permuted, with the updating of $\bm{d}_{1:n}$.
%However, as we will see, in practice it is not necessary to permute the data set at all and we no longer require acceleration steps to improve the mixing properties.

%We start by noting that, since data points are exchangeable, at any step of any iteration  of the OAS we can apply a permutation $\rho = (\rho(1),\ldots,\rho(n))$ of $\{1,...,n\}$ to the data obtaining $(y_{\rho(1)},\ldots,y_{\rho(n)})$
%%as long as \blue{this permutation} satisfies
%such that
%\begin{equation}\label{eq:y_exch}
%\pi(y_1,\ldots,y_n) = \pi(y_{\rho(1)},\ldots,y_{\rho(n)}).
%\end{equation}
%To keep consistent the parametrization established in the OAS, all the other random terms must be relabelled in order of appearance in the permuted data set.
%As far as the other random terms are concerned, they need to be relabelled in order of appearance in the permuted data set
%By definition of exchangeability, not all permutations are valid: for example if $\rho$ depends on the values of $\bm{y}_{1:n}$, \eqref{eq:y_exch} does not hold.
%} \red{\it[xx what the valid ones then? How can we justify that a uniform at random, or a deterministic like the one we use here, are legit? xx]}

%We can take advantage of this in the following way. 
Specifically, before updating each ordered allocation variable, we apply the permutation that sends the first element to the last place, this is the permutation $\rho = (2,\ldots,n,1)$, and relabel every other random term accordingly. By doing so sequentially, instead of updating $d_i$ corresponding to $\bm{y}_{1:n}$, we update $d'_n$ corresponding to $\bm{y}'_{1:n} = (y_{i+1},\ldots,y_n,y_1,\ldots,y_i)$. Notice that both $d_i$ and $d'_n$ indicate  the mixture component to which $y_i = y'_n$ is associated. It follows easily from \eqref{eq:d_post_prev} that
\begin{equation}\label{eq:d'_post0}
  \Prob(\,d'_n = d\mid \rest\,) \propto \begin{cases}
\tp'_d\, g(y'_n\mid \tx'_d) & \text{ for } d \in \{1,\ldots, k\},\\
\l(1-\sum_{l=1}^{k}\tp'_l\r)g(y'_n\mid \tx'_d) & \text{ for } d = k+1,
\end{cases}
\end{equation}
%and $\pi(\,d'_n = d\mid \rest\,) = 0$ if $d \not\in \{1,\ldots,k+1\}$
Here $k =\max\{\bm{d}'_{1:n-1}\}$, where $\bm{d}'_{1:n-1}$ are the ordered allocation variables corresponding to $\bm{y}'_{1:n-1}=(y_{i+1},\ldots,y_n,y_1,\ldots,y_{i-1})$. Accordingly, $\bm{\tx}'_{1:k} = (\tx'_1,\ldots,\tx'_k)$, and $\bm{\tp}'_{1:k} = (\tp'_1,\ldots,\tp'_k)$ are the observed component parameters and weights in the order in which they are discovered by $\bm{y}'_{1:n-1}$. With respect to the current rearrangement, $\tx'_{k+1}$ and $\tp'_{k+1}$ refer to the next component parameter and weight to be discovered.
Note that in order to sample $d'_n$ from \eqref{eq:d'_post0} we require to condition on the values of  $\bm{\tx}'_{1:k}$,  $\bm{\tp}'_{1:k}$ and well as $\tx'_{k+1}$ but not on $\tp'_{k+1}$. Being that unobserved component parameters are updated from the prior, cf. \eqref{eq:x_post} in Appendix \ref{app:OAS}, at this stage we can simply sample $\tx'_{k+1} \sim \nu$ before updating $d'_n$. 
Now, after updating $d'_n$, two scenarios arise.
%\red{\it[xx so the prime notation $\tx'_{k+1}$ is not needed, as its law does not depend on the reordering of the $\tx_1,\ldots,\tx_k$. Why don't we drop it? xx]}
%\begin{enumerate}[(i)]
%\item
In the first scenario, $d'_n \leq k$, this means that $y'_n = y_i$ is assigned to an already observed component and we can simply dismiss $\tx'_{k+1}$ as it remains undiscovered.
%since $\tx'_{k+1}$ was not discovered we dismiss it.
%\item
Otherwise, $d'_n = k+1$ meaning that $y'_n = y_i$ is assigned to a new component with parameter $\tx'_{k+1}$, not discovered by $\bm{y}'_{1:n-1}$, in this case we must also update the corresponding weight $\tp'_{k+1}$.
%Also in this case we keep the prime notation in $\tp'_{k+1}$ even tough its law does not depend on the permutation of the data.
%To see why, note that
Now, $\bm{\tp}' = (\tp'_j)_{j=1}^{\infty}$ are the weights in order of appearance with respect to the permutation of an exchangeable sequence, hence $\bm{\tp}'$ remains ISBP, i.e. $\pi(\bm{\tp}') = \pi(\bm{\tp})$. At this stage we further have $n'_{k+1} = |\{i: d'_n = k+1\}| = 1$, hence, it follows from \eqref{eq:like_sb} and \eqref{eq:prior_sb} that the full conditional of $\tp'_{k+1}$ depends only on $\bm{\tp}'_{1:k}$, 
\begin{equation}\label{eq:p'_k+1}
  \pi(\tp'_{k+1}\mid \cdots) = \pi(\tp'_{k+1}\mid\bm{\tp}'_{1:k}),
\end{equation}
and that this conditional distribution corresponds to the prior law $\pi(\bm{\tp'}) =\pi(\bm{\tp})$ of the weights in size-biased order. For illustration, in the Pitman-Yor case, cf. centered equation after \eqref{eq:prior_sb}, one can sample $\tp'_{k+1}$ from \eqref{eq:p'_k+1} by setting
$$
\tp'_{k+1} = v_{k+1}\Big(1-\sum\nolimits_{j=1}^{k}\tp'_j\Big)$$
where $v_{k+1} \sim \Be(1-\sigma,\beta+(k+1)\sigma)$, also see \eqref{eq:v_post_sb} in Appendix \ref{app:OAS}.

The main advantage of updating the component to which $y_i$ is allocated through $d'_n$ as in \eqref{eq:d'_post0} instead of sampling $d_i$ from \eqref{eq:d_post_prev} is that the admissible moves for $d'_n$ are precisely all observed components plus a new one, whilst the ones for $d_i$ are potentially much less. Seemingly, the disadvantage with respect to the original OAS is that instead of permuting the data set once per iteration in the acceleration step, we permute data points and relabel the rest of the random elements each time we update an allocation variable, thus $n$ times for each iteration. However, in practice it is not necessary to do so even once. To explain this, notice that \eqref{eq:d'_post0} is invariant under permutations of the pairs $(\tx'_1,\tp'_1),\ldots(\tx'_k,\tp'_k)$. This is, if $\bm{x}^*_{1:k} = (x^*_1,\ldots,x^*_k)$ 
%= \big(\tx'_{\rho(1)},\ldots,\tx'_{\rho(k)}\big)$ 
is any rearrangement of $\bm{\tx}'_{1:k}$ and $\bm{p}^*_{1:k} = (p^*_1,\ldots,p^*_k)$ 
%= \big(\tp'_{\rho(1)},\ldots,\tp'_{\rho(k)}\big)$ 
is the correspondent relabelling of $\bm{\tp}'_{1:k}$, then sampling $d'_n$ from \eqref{eq:d'_post0} is equivalent to sample $c_i$ from
\begin{equation}\label{eq:c_post0}
  \Prob(\,c_i = c\mid \rest\,) \propto \begin{cases}
p^*_c\, g(y_i\mid x^*_c) & \text{ if } c \in \{1,\ldots,k\},\\
\l(1-\sum_{l=1}^{k}p^*_l\r) g(y_i\mid x^*_c)& \text{ if } c = k+1,
\end{cases}
\end{equation}
where $x^*_{k+1} \sim \nu$, and later set 
%For  $\sigma_1,\ldots,\sigma_{k_n}$ the distinct values of $\bm{c}_{1:n}$ in order of appearance, set $d_i = j$ when $c_i = \sigma_j$.
%For $i \in \{1,\ldots,n\}$, set $d_i = j$ if and only if $c_i$ is equal to the $j$th distinct value to appear in $\bm{c}_{1:n}$.
$d'_{n} = k+1$ if $c_i = k+1$ or set $d'_n = j$ when $y_i = y'_n$ has been associated to the $j$th component in order of appearance in $\bm{y}'_{1:n}$, i.e. $x^*_{c_i} = \tx'_j$ and $p^*_{c_i} = \tp'_j$, for some $j \leq k$.
%It is easy to see that the least element order is satisfied with respect to the ordering of the data resulting from the composition of $\rho = (2,\ldots,n,1)$ $i$ times.
Furthermore, we also note that since the weights in order of appearance are ISBP, the conditional distribution in \eqref{eq:p'_k+1} is symmetric in $\bm{\tp}'_{1:k} = (\tp'_1,\ldots,\tp'_k)$, see Proposition \ref{prop:p_k+1} for a proof, and the Pitman-Yor case detailed before for an example.
%Cf. the Pitman-Yor case detailed before.
%model sampling $\tp'_{k+1}$ from \eqref{eq:p'_k+1} can be done by sampling $v_{k+1} \sim \Be(1-\sigma,\beta+(k+1)\sigma)$ and set $\tp'_{k+1} = v_{k+1}(1-\sum_{l=1}^k\tp'_l) = v_{k+1}(1-\sum_{l=1}^k p^*_l)$ cf. \eqref{eq:v_post_sb} in Appendix \ref{app:OAS}.
Thus, whenever we require to sample the weight, $\tp'_{k+1}$, of a newly discovered component, we can sample it from \eqref{eq:p'_k+1}
%  $\pi(\,\cdot\,\mid \bm{p}^*_{1:k})  = \pi(\,\cdot\,\mid\bm{\tp}_{1:k})$,
%that is conditional on $\bm{\tp}_{1:k} = (\tp_1,\ldots,\tp_k)$ in any ordering.
conditioning on the permutation $\bm{p}^*_{1:k} = (p^*_1,\ldots,p^*_{k})$ of $\bm{\tp}'_{1:k} = (\tp'_1,\ldots,\tp'_k)$.
%Henceforth we will refer to the newly discovered weight as $p^*_{k+1}$.}

With these considerations, we get that as long as $\bm{x}^*_{1:k}$ and $\bm{p}^*_{1:k}$ are the component parameters and weights discovered by $\bm{y}_{-i} = (y_1,\ldots,y_{i-1},y_{i+1},\ldots,y_n)$, the order in which they were discovered is irrelevant for both, the updating of the component to which $y_i$ belongs (now indicated by $c_i$), and the updating of the newly discovered component parameters and weights (henceforth denoted by $x^*_{k+1}$ and $p^*_{k+1}$ for consistency). This in turn yields that instead of permuting data points and relabel in order of appearance each time an ordered allocation variable is  updated, we can update the whole collection $\bm{d}_{1:n}$ via \textit{unordered} allocation variables $\bm{c}_{1:n} = (c_1,\ldots,c_n)$, by sequentially sampling from \eqref{eq:c_post0}. In other words, we first update the clustering of data points through $\bm{c}_{1:n}$, once this collection is fully updated, we obtain $\bm{d}_{1:n}$ by defining $d_i = j$ if and only if $c_i$ coincides with $j$th distinct value to appear in $\bm{c}_{1:n}$, and proceed with the OAS as derived by \cite{Deb:Gil:23}. See Algorithm \ref{alg:EOAS_PY} for the Pitman-Yor prior case. 
%\\
%\red{\it[xx introduce the permutation of $1,\ldots,k$ determined by the distinct values of the unordered allocation variables $c_i$'s here, adding \comillas{this is equivalent to set $d_i = j$ when $c_i = \sigma_j$ for $\sigma_1,\ldots,\sigma_{k_n}$ the distinct values of $\bm{c}_{1:n}$ in order of appearance with respect to $\bm{y}_{1:n}$} in the current version we consistently use $\sigma_1,\ldots,\sigma_{k_n}$ in all algorithms, rather than former $\rho_1,\ldots,\rho_{k_n}$ that were creating confusion with the notation when we deal with the updating of the $\alpha$'s. xx]}
%\\
It is very important to highlight, however, that one cannot simply replace the weights in order of appearance with the weights in any arbitrary ordering tout court.
%The first $k$ elements $\bm{p}^*_{1:k}$ have to be the first $k$ weights to be discovered
Indeed, elements of $\bm{p}^*_{1:k}$ have to be the first $k$ weights to be discovered (in some ordering) by
$\bm{y}'_{1:n-1}=(y_{i+1},\ldots,y_n,y_1,\ldots,y_{i-1})$,
thus these are the first $k$  entries of a sequence, $\bm{\tp}'$ that a priori is ISBP. In particular this guarantees that the conditional distribution in \eqref{eq:p'_k+1} is symmetric in $\bm{p}'_{1:k}=(\tp_1',\ldots,\tp'_k)$ and thus we can equivalently condition on the relabelling $\bm{p}^*_{1:k} =  (p^*_1,\ldots,p^*_{k})$ for the updating of $p^*_{k+1}$. In general this is not true for weights that are not ISBP a priori.
%\\
%\red{\it[xx remove \comillas{, it is necessary that these weight are \red{the first $k$} in size-biased random order} do you mean instead in order of appearance instead? xx]}
%\\
%\red{\it[xx former paragraph \comillas{Hence, even if ordered allocation variables are updated via the (not necessarily ordered) allocation variables $\bm{c}_{1:n} = (c_1,\ldots,c_n)$, it is necessary to later compute $\bm{d}_{1:n}$ by defining $d_i = j$ if and only if $c_i$ coincides with $j$th distinct value to appear in $\bm{c}_{1:n}$, and proceed with the OAS as derived by \cite{Deb:Gil:23}} is not clear to me. Why \comillas{Hence}? xx]}
%\\
%Hereinafter, the variant of the OAS described here is called \emph{Efficient Ordered Allocation Sampler} (EOAS). \red{\it[xx I prefer e-OAS, I would like people to remember the original acronym OAS xx]} 

\begin{algorithm}[tb]
\caption{OAS - Pitman-Yor case}
\label{alg:EOAS_PY}
\begin{small}
Let $\bm{d}_{1:n}$ be the current values of the ordered allocation variables, $k_n = \max\{\bm{d}_{1:n}\}$ and $D_j = \{i \leq n: d_i = j\}$ for $j\leq k_n$. Also let
$\bm{\tx}_{1:k_n}$, $\bm{\tp}_{1:k_n}$ be the component parameters and weights in order of appearance in $\bm{y}_{1:n}$ through $\bm{d}_{1:n}$.
\begin{enumerate}
\item 
%Update $\bm{d}_{1:n}$ via $\bm{c}_{1:n}$. 
Set $c_i = d_i$, for $i\in \{1,\ldots, n\}$, and set $x^*_j = \tx_j$, $p^*_j = \tp_j$ and $C_j = D_j$ for $j \leq k_n$. Repeat the following steps for $i\in \{1,\ldots,n\}$:
\begin{enumerate}[(i)]
\item If $C_{c_i} = \{i\}$, set $k = k_n-1$, for each $l$ such that $c_l > c_i$ redefine $c_l = c_l-1$, delete $x^*_{c_i}$ and $p^*_{c_i}$, and set $x^*_j = x^{*}_{j+1}$ and $p^*_j = p^*_{j+1}$, for $c_i \leq j \leq k$. Otherwise, simply set $k = k_n$. In any case, $\bm{x}^*_{1:{k}} = (x^*_1,\ldots,x^*_{k})$ and $\bm{p}^*_{1:{k}} = (p^*_1,\ldots,p^*_{k})$ are the component parameters and weights discovered by $\bm{y}_{-i} = (y_1,\ldots,y_{i-1},y_{i+1},\ldots,y_n)$, and $\bm{c}_{-i} = (c_1,\ldots,c_{i-1},c_{i+1},\ldots,c_n)$ are the corresponding allocation variables.
\item Sample $x^*_{k+1} \sim \nu$ and $c_i$ from \eqref{eq:c_post0}. If $c_i = k+1$, set $k_n = k+1$, sample $v_{k+1} \sim\Be(1-\sigma,\beta+(k+1)\sigma)$, and set $p^*_{k+1} = v\big(1-\sum_{l=1}^{k}p^*_l\big)$. If $c_i \leq k$, set $k_n = k$ and dismiss $x^*_{k+1}$.
\end{enumerate}
%Set $d_i = j$ when $c_i = \sigma_j$ for $\sigma_1,\ldots,\sigma_{k_n}$ the distinct values of $\bm{c}_{1:n}$ in order of appearance.
For  the distinct values,  $\sigma_1,\ldots,\sigma_{k_n}$, of $\bm{c}_{1:n}$ in order of appearance, set $d_i = j$ when $c_i = \sigma_j$, for $i\in \{1,\ldots,n\}$.
%For $i \in \{1,\ldots,n\}$, set $d_i = j$ if and only if $c_i$ is equal to the $j$th distinct value to appear in $\bm{c}_{1:n}$.
Also set $D_j = \{i \leq n: d_i = j\}$ and $n_j = |D_j|$ for $j\in \{1,\ldots,k_n\}$.
\item 
Sample  $\bm{\tx}_{1:k_n}$ from  $\pi(\,\bm{\tx}_{1:k_n}\mid \cdots) \propto \prod_{j = 1}^{k_n}\prod_{i \in D_j} g(y_i\mid \tx_j) \times \nu(\tx_j)$, cf. \eqref{eq:x_post}.
\item Sample  $\bm{\tp}_{1:k_n}$  from $\pi(\,\bm{\tp}_{1:k_n}\mid\rest) \propto \prod_{j=1}^{k_n}\tp_j^{n_j-1}
  \big(1-\sum_{l=1}^{j-1}\tp_l\big)\pi(\bm{\tp}_{1:k_n})$, cf. \eqref{eq:p_post}. That is, sample $v_j \ind  \mathsf{Be}\big(n_j-\sigma,\sum_{l>j}n_l+\beta+j\sigma\big)$ and set $\tp_j = v_j\prod_{l=1}^{j-1}(1-v_l)$, for $j=1,\ldots, k_n$, see \eqref{eq:v_post_sb}.
\end{enumerate}
\end{small}
\end{algorithm}

Comparing Algorithm \ref{alg:EOAS_PY} with the marginal algorithm \citep[cf.][Algorithm 8]{Nea:00}, there are only a few differences.
%\red{\it[xx I do not think we need an itemized list, actually the second and third one can be merged into a single one xx]}
First, in the marginal algorithm, the updating of $c_i$ requires the marginal predictive distribution $\Prob(\theta_{n} \in \cdot\mid \bm{\theta}_{1:n-1})$ (equivalent to $\Prob(c_{n} \in \cdot\mid \bm{c}_{1:n-1})$).
%\\
%\red{\it[xx misleading that we are back to the notation $\Prob(\,\cdot\,)$. See comment after equation \eqref{eq:like_sb} xx]}
%\\
Specifically, for Pitman-Yor models, one has
\begin{equation}\label{eq:marg_pred}
\Prob(\theta_{n} \in \cdot\mid \bm{\theta}_{1:n-1}) = \sum_{j=1}^{k_{n-1}}\frac{m_j-\sigma}{\beta +n}\delta_{x^*_j} + \frac{\beta + k\sigma}{\beta + n}\nu,
\end{equation}
where $x^*_1,\ldots,x^*_{k_{n-1}}$ are the distinct values that $\bm{\theta}_{1:n-1}$ exhibits (in no particular order) and $m_j = |\{i \leq n-1: m_j = x^*_j\}|$. This in turn yields
\begin{equation}\label{eq:c_post0_mar}
  \Prob(\,c_i = c\mid \rest\,) \propto 
  \begin{cases}
  (m_c-\sigma)\, g(y_i\mid x^*_c) & \text{ if } c \in \{1,\ldots,k\},\\
  (\beta+ k\sigma)g(y_i\mid x^*_c)& \text{ if } c = k+1,
  \end{cases}
\end{equation}
where $\bm{x}^*_{1:k} = (x^*_1,\ldots,x^*_k)$ are the observed component parameters, $x^*_{k+1} \sim \nu$, $c_i = c$ indicates that the $i$th data point, $y_i$, is associated to the component parameter $x^*_c$, and $m_c = |\{l \neq i: c_l = c\}|$ is the number observations in $\bm{y}_{-i} = (y_1,\ldots,y_i,y_{i+1},y_n)$ currently associated to $x^*_c$. In the OAS, by including the weights, we replace $(m_c-\sigma)$ and $(\beta+k\sigma)$  by $p^*_c$ and $1-\sum_{j=1}^kp^*_j$, respectively, cf. \eqref{eq:c_post0}. Also compare \emph{iii} in Theorem \ref{theo:rep_sss} with \eqref{eq:marg_pred}.
Secondly, in comparison to marginal samplers, the OAS features an additional step where the observed weights, $\bm{\tp}_{1:k_n}$, are updated. Recall that in marginal samplers these have been integrated out. In particular, when a new component is discovered, in the OAS we must sample the corresponding weight $p^*_{k+1}$. In marginal samplers this is not necessary, instead we simply set $m_{k+1} = 1$ when we update $c_i$ successively.

%Thus the EOAS can be regarded as a conditional version of a marginal sampler. \red{\it[xx why \comillas{Thus}? xx]} Due to this, \red{\it[xx due to what? xx]} it is expected that marginal samplers will have an edge in terms of mixing properties over the EOAS. The latter, in turn, should perform better than other conditional samplers. 
%\noindent \dots \\ \dots

Despite the similarities, the new OAS remains a conditional sampler, namely it evolves in a larger state space when compared to marginal samplers due to the inclusion of observed weights. Hence, it is expected that marginal samplers will have an edge in terms of mixing properties. The OAS, in turn, is expected to perform better than other conditional samplers since it now mixes in the space of partitions. The most important advantage of the OAS over marginal samplers is that it can be adapted to a wider range of mixing priors, in that it does not require the predictive distribution $\Prob(\theta_{n} \in \cdot\mid \bm{\theta}_{1:n-1})$ in closed form. In this it shares the same wide applicability of other conditional sampler such as the slice sampler. In general the availability of the predictive distribution $\Prob(\theta_{n} \in \cdot\mid \bm{\theta}_{1:n-1})$ in closed form corresponds to the availability of the law of the weights in size-biased order. So a key feature of original OAS consists in modeling $\bm{\tp}_{1:k_n}$ via the weights in any arbitrary order, $\bm{p}$ and the indexes, $\bm{\alpha}_{1:k_n} = (\alpha_1,\ldots,\alpha_{k_n})$ satisfying \eqref{eq:size-biased_pick}, that indicate which elements of $\bm{p}$ are first ones to be discovered by $\bm{y}_{1:n}$. In particular, the weights $\bm{\tp}'_{1:k}$ in \eqref{eq:d'_post0} will be given by $\tp'_j = p_{\alpha'_j}$ where the rearranged indexes $\bm{\alpha}'_{1:k} = (\alpha'_1,\ldots,\alpha'_k)$ now specify which weights in $\bm{p}$ are the first ones to be discovered by $\bm{y}'_{1:n-1}$. Note that because data is exchangeable, the prior law of $\bm{\alpha}'_{1:k} = (\alpha'_1,\ldots,\alpha'_k)$
coincides with that in \eqref{eq:size-biased_pick}. In this context, any permutation say $\bm{p}^*_{1:k}$ of $\bm{\tp}'_{1:k}$ can be used by considering the respective relabelling $\bm{\alpha}^*_{1:k} = (\alpha^*_1,\ldots,\alpha^*_k)$ of  $\bm{\alpha}'_{1:k}$ and setting $p^*_j = p_{\alpha^*_j}$. As for the updating of a newly discovered weight, this can be done by defining $p^*_{k+1} = p_{\alpha^*_{k+1}}$ where $\alpha^*_{k+1}$ is sampled from  \eqref{eq:size-biased_pick} conditioning on $\bm{\alpha}^*_{1:k}$, and all the required weights up to $p_{\alpha^*_{k+1}}$ are sampled from the appropriate conditional distribution corresponding to the prior law $\pi(\bm{p})$. This is detailed in Appendices \ref{app:OAS} and \ref{app:SBW}, see in particular Proposition \ref{prop:pa_k+1} and Remark \ref{rem:pa_k+1} therein.
%and the required extra weights, $p_l$, are sampled from $\pi(p_l\mid \bm{p}_{1:l-1})$ corresponding to the prior law $\pi(\bm{p})$. Further details can can be consulted in Appendices \ref{app:OAS} and \ref{app:SBW}.
In Algorithm \ref{alg:EOAS} we describe one iteration of the most general version of the OAS.

\begin{algorithm}[tb]
\caption{OAS - general case}\label{alg:EOAS}
\begin{small}
Let $\bm{d}_{1:n}$ be the current values of the ordered allocation variables, $k_n = \max\{\bm{d}_{1:n}\}$ and $D_j = \{i \leq n: d_i = j\}$ for $j\leq k_n$. Also let
$\bm{\tx}_{1:k_n}$, $\bm{\tp}_{1:k_n}$ be the component parameters and weights in order of appearance in $\bm{y}_{1:n}$ through $\bm{d}_{1:n}$, respectively. Here, $\bm{\tp}_{1:k_n}$ is determined by $\bm{p}_{1:\overline{\bm{\alpha}}_{k_n}}$ and $\bm{\alpha}_{1:k_n}$ through $\tp_j = p_{\alpha_j}$, where $\overline{\bm{\alpha}}_{k_n} = \max\{\bm{\alpha}_{1:k_n}\}$.
\begin{enumerate}
\item Update $\bm{d}_{1:n}$ via $\bm{c}_{1:n}$. Set $c_i = d_i$, for $i \leq n$, and set $x^*_j = \tx_j$, $\alpha^*_j = \alpha_j$ and $C_j = D_j$ for $j \leq k_n$. Repeat the following steps for $i \in \{1,\ldots,n\}$:
\begin{enumerate}[(i)]
\item If $C_{c_i} = \{i\}$, set $k = k_n-1$, for each $l$ such that $c_l > c_i$ redefine $c_l = c_l-1$, delete $x^*_{c_i}$ and $\alpha^*_{c_i}$, set $x^*_j = x^{*}_{j+1}$ and $\alpha^*_j = \alpha^*_{j+1}$, for $c_i \leq j \leq k$, and delete $p_j$ with $j > \max\{\bm{\alpha}^*_{1:k}\}$. Otherwise, simply set $k = k_n$. In any case, $\bm{x}^*_{1:{k}} = (x^*_1,\ldots,x^*_{k})$ and $\bm{p}^*_{1:{k}} = (p^*_1,\ldots,p^*_{k})$, with $p^*_j = p_{\alpha^*_j}$, are the component parameters and weights discovered by $\bm{y}_{-i} = (y_1,\ldots,y_{i-1},y_{i+1},\ldots,y_n)$, and $\bm{c}_{-i} = (c_1,\ldots,c_{i-1},c_{i+1},\ldots,c_n)$ are the corresponding allocation variables.
\item Sample $x^*_{k+1} \sim \nu$ and $c_i$ from \eqref{eq:c_post0}.
If $c_i = k+1$, sample $\alpha^*_{k+1}$ from
$
\pi(\alpha^*_{k+1}\mid \bm{\alpha}^{*}_{1:k},\bm{p}) \propto p_{\alpha^*_{k+1}}\Ind_{\{\alpha^*_{k+1} \not\in \bm{\alpha}^{*}_{1:k}\}}
$
as in \eqref{eq:size-biased_pick}, as well as all the required weights up to $p_{\alpha^*_{k+1}}$, if any, from the appropriate conditional distribution corresponding to the prior law $\pi(\bm{p})$ (cf. \eqref{eq:post_p_a}, Proposition \ref{prop:pa_k+1} and Remark \ref{rem:pa_k+1} in Appendices \ref{app:OAS} and \ref{app:SBW}). Also set $k_n = k+1$.
Otherwise, if $c_i \leq k$, set $k_n = k$ and dismiss $x^*_{k+1}$.
\end{enumerate}
 For  the distinct values,  $\sigma_1,\ldots,\sigma_{k_n}$, of $\bm{c}_{1:n}$ in order of appearance, set $d_i = j$ when $c_i = \sigma_j$, for $i\in \{1,\ldots,n\}$. Then set $\alpha_j = \alpha^*_{\sigma_j}$, $D_j = \{i \leq n: d_i = j\}$ and $n_j = |D_j|$ for $j\in \{1,\ldots,k_n\}$.
\item Sample $\bm{\tx}_{1:k_n}$ from  $\pi(\,\bm{\tx}_{1:k_n}\mid \cdots) \propto \prod_{j = 1}^{k_n}\prod_{i \in D_j} g(y_i\mid \tx_j) \times \nu(\tx_j)$, cf. \eqref{eq:x_post}.
\item Update $\bm{\tp}_{1:k_n}$ via  $\bm{\alpha}_{1:k_n}$ and $\bm{p}_{1:\overline{\bm{\alpha}}_{k_n}}$:
\begin{itemize}
\item[(i)] Sample $\rho = (\rho(1),\ldots,\rho(k_n))$ from
$
\pi(\rho) \propto \prod_{j=1}^{k_n}w_{j,\rho(j)}\Ind_{\{\rho \in \mathcal{S}_{k_n}\}},
$
where $w_{j,l} = p_{\alpha_l}^{n_j}$ and $\mathcal{S}_{k_n}$ is the space of all permutations of $\{1,\ldots,k_n\}$, cf. \eqref{eq:pi_rho} in Appendix \ref{app:OAS}. Later, apply $\rho$ to the indexes of $\bm{\alpha}_{1:k_n}$ obtaining the updated value $(\alpha_{\rho(1)},\ldots,\alpha_{\rho(k_n)})$ of $\bm{\alpha}_{1:k_n}$.
\item[(i*)] Optionally, for $j \leq k_n$, sample $\alpha_j$ from
$
q(\alpha_j) \propto p^{n_j}_{\alpha_j}\Ind_{\{\alpha_j \neq \alpha_l, \, \forall\, l \leq k_n, l \neq j\}}
$
as in \eqref{eq:a_j_acc_step}. To do so, perform Algorithm \ref{alg:alpha} in Appendix \ref{app:OAS}.
\item[(ii)] Sample  $\bm{p}_{1:\overline{\bm{\alpha}}_{k_n}}$ from $\pi(\,\bm{p}_{1:\overline{\bm{\alpha}}_{k_n}}\mid \cdots) \propto \prod_{j=1}^{\overline{\bm{\alpha}}_{k_n}}p_{j}^{r_j} \pi(\bm{p}_{1:\overline{\bm{\alpha}}_{k_n}})$, with $r_j = \sum_{l=1}^{k_n}n_l\Ind_{\{\alpha_l = j\}}$, cf. \eqref{eq:post_p}.
\end{itemize}
\end{enumerate}
\end{small}
\end{algorithm}

We discuss next the three issues raised in Section \ref{sec:OAS}. 
(a) As we have discussed, updating $d_i$ via sampling $c_i$ from \eqref{eq:c_post0} is equivalent to permute the dataset obtaining $\bm{y}'_{1:n}$, with $y'_n = y_i$, relabel all random elements in order of appearance, and update $d'_{n}$ from \eqref{eq:d'_post0}. Thus we have that the initial ordering of data points should not affect the mixing properties of the new OAS and this version of the algorithm does not requires the acceleration step in \cite{Deb:Gil:23}.
(b) In the new variant of the OAS,  when updating allocation variables, any data point $y_i$ can either be attached to a discovered component or a new one exactly as in marginal samplers. Hence any number of components can be created or destroyed at each iteration, and thus, in terms of the mixing of $k_n$ it is expected that the sampler performs better than the original version in \cite{Deb:Gil:23}.
(c) Since $\bm{d}_{1:n}$ are now updated via unordered allocation variables $\bm{c}_{1:n}$, it is no longer necessary to compute the admissible moves $\D_i$ for $d_i$. This makes the sampler a lot easier to program and faster to run when compared to the the original version in \cite{Deb:Gil:23}.

\section{Empirical comparison of Gibbs samplers}

In this section we run a simulation study to test the performances of the OAS with four different species sampling mixing priors and four synthetic datasets. We compare the new version of the OAS with its original version in \cite{Deb:Gil:23}, together with a marginal and a conditional algorithm. 
For quick reference, in this section we refer to the original version as oOAS, keeping OAS for the new version. The first dataset is the \texttt{galaxy} data which includes the velocities of 82 galaxies diverging away from our galaxy. Each of the remaining three datasets consist of $100$ draws from a mixture of Gaussian distributions. 
The \texttt{leptokurtic} data was generated from the mixture $0.67\mathsf{N}(0,1) + 0.33\mathsf{N}(0,0.25^2)$, the \texttt{bimodal} data from $0.5\mathsf{N}(-1,0.5^2) + 0.5\mathsf{N}(1,0.5^2)$. Finally for the \texttt{mix} data we considered a mixture of a leptokurtic and a bimodal distribution, namely $0.5[0.6\mathsf{N}(-2,1.25^2) + 0.4\mathsf{N}(-2,0.25^2)] + 0.5[0.5\mathsf{N}(1.5,0.75^2) + 0.5\mathsf{N}(4,0.75^2)]$. The mixing priors were chosen so that two of them have tractable size-biased ordered weights, $\bm{\tp}$, these are a Dirichlet process (DP) and a Pitman-Yor process (PY). The other two mixing priors are (an instance of) the Exchangeable stick-breaking process \citep[ESB][]{Gil:Men:21} and the Geometric process \citep[GP][]{Fue:etal:10}. These entail weights $\bm{p}$ that are not in size-biased order.
%have tractable weights, $\bm{p}$, in a certain ordering, but this arrangement is not the one that is ISBP. 
In all cases the mixing kernel in \eqref{eq:mix} is Gaussian, $g(y\mid x) = \mathsf{N}(y \mid \mu,\tau^{-1})$ for $x=(\mu,\tau)$, and to attain conjugacy between $g(y\mid \mu,\tau)$ and the base measure $\nu(\ddr \mu,\ddr \tau)$, we set the latter to be the Gaussian--Gamma distribution 
  $\mathsf{N}(\mu \mid \mu_0,(\lambda_0\tau)^{-1})\mathsf{Ga}(\tau\mid a_0,b_0)$. 
The hyperparameters were fixed to $\mu_0 = n^{-1}\sum_{i=1}^{n}y_i$, $\lambda_0 = 1/100$, and $a_0 = b_0 = 0.5$.

The marginal Algorithm 8 of \cite{Nea:00} was used for the DP and PY priors. No marginal sampler is available for the ESB and GP priors since they do not have a tractable predictive distribution $\Prob(\theta_{n} \in \cdot\mid \bm{\theta}_{1:n-1})$. For the conditional sampler, we implemented the (dependent) slice-efficient sampler of \cite{Kall:etal:11} for all four mixing priors. As for the oOAS and the OAS, we choose the versions that model directly size-biased ordered weights, $\bm{\tp}$, for the DP and the PY, the versions that model $\bm{\tp}$ via $\bm{p}$ and $\bm{\alpha}$ for the ESB and the GP. All samplers were ran for $N = 2\times 10 ^6$ iterations, after a burn-in of $B = 10^5$ iterations. To monitor the algorithmic performance we considered the \emph{integrated autocorrelation time} \citep[IAT, ][]{Sok:97} $\tau = 1/2 + \sum_{l=1}^{\infty} \rho_l$, of the number of occupied components, $k_n$, and the deviance of the density, $D_v = -2\sum_{i=1}^n\log \sum_{j=1}^{\infty} \frac{n_j}{n} g(y_i\mid x_j)$. Here $n_j$ is the number of data points associated to $g(y_i\mid x_j)$, and $\rho_l$ refers to the $l$--lag autocorrelation of the monitored chain. As explained by \cite{Sok:97}, if $\hat{f}$ is the Monte Carlo estimator of a quantity of interest, $f$, then $Var(\hat{f}) \approx 2 \tau \times Var(f)N^{-1}$, where $Var(f)$ refers to the marginal variance of $f$. Hence smaller values of $\tau$ correspond to more efficient samplers. As done by \cite{Kall:etal:11} and \cite{Deb:Gil:23}, the IAT is computed as $\hat{\tau} = 1/2 +\sum_{l=1}^{C-1} \hat{\rho}_l$, where $\hat{\rho}_l$ is the estimated autocorrelation at lag $l$ and $C = \min\{l: \hat{\rho}_l < 2/\sqrt{N}\}$. The results are reported in Table \ref{tab:IAT}.

\begin{table}
\begin{scriptsize}
\centering
\begin{tabular}{|c | C{1.2cm}C{1.2cm} | C{1.2cm}C{1.2cm} | C{1.2cm}C{1.2cm} | C{1.2cm}C{1.3cm}|}
%\multicolumn{9}{c}{\texttt{galaxy} data} \\
\multicolumn{1}{c}{\texttt{galaxy}}
& \multicolumn{2}{c}{DP}
& \multicolumn{2}{c}{PY}
& \multicolumn{2}{c}{ESB}
& \multicolumn{2}{c}{GP} \\ \hline
& $D_v$ & $k_n$ & $D_v$ & $k_n$ & $D_v$ & $k_n$ & $D_v$ & $k_n$ \\ \hline
Mar & 12.44(0.25) & 14.08(0.34) & 11.42(0.19) & 10.60(0.18) &  --- & --- &  --- & ---  \\
OAS & 19.43(0.49) & 22.55(0.60) & 17.86(0.37) & 16.77(0.37) & 22.99(0.64) & 31.95(1.01) & 17.89(0.59) & 29.71(0.90) \\
oOAS & 25.24(0.66) & 32.69(0.80) & 17.18(0.42) & 34.25(0.90) & 25.65(0.83) & 63.54(2.59) & 18.21(0.70) & 63.62(1.94)\\
Cond. & 130.42(8.7) & 195.5(13.2) & 69.46(3.85) & 99.15(5.27) & 295.7(34.7) & 556.3(62.2) & 97.48(12.9) & 17080(8684) \\ \hline
\end{tabular}

\vspace{0.5cm}

\begin{tabular}{|c | C{1.2cm}C{1.2cm} | C{1.2cm}C{1.2cm} | C{1.2cm}C{1.2cm} | C{1.2cm}C{1.2cm}|}
%\multicolumn{9}{c}{\texttt{galaxy} data} \\
\multicolumn{1}{c}{\texttt{lepto}}
& \multicolumn{2}{c}{DP}
& \multicolumn{2}{c}{PY}
& \multicolumn{2}{c}{ESB}
& \multicolumn{2}{c}{GP} \\ \hline
& $D_v$ & $k_n$ & $D_v$ & $k_n$ & $D_v$ & $k_n$ & $D_v$ & $k_n$ \\ \hline
Mar & 11.74(0.18) & 6.89(0.08) & 14.46(0.33) & 4.40(0.04) &  --- & --- &  --- & ---  \\
OAS & 14.59(0.28) & 10.50(0.14) & 17.67(0.34) & 6.77(0.08) & 31.23(0.80) & 13.07(0.31) & 10.56(0.15) & 12.32(0.22) \\
oOAS & 14.13(0.23) & 14.72(0.26) & 16.34(0.30) & 12.57(0.18) & 31.23(0.75) & 16.40(0.37) & 13.08(0.22) & 17.50(0.35) \\
Cond. & 57.62(4.13) & 52.97(0.28) & 88.85(8.60) & 48.79(2.30) & 264.0(34.7) & 60.51(3.82) & 43.55(1.36) & 50.74(1.56) \\ \hline
\end{tabular}

\vspace{0.5cm}

\begin{tabular}{|c | C{1.2cm}C{1.2cm} | C{1.2cm}C{1.2cm} | C{1.2cm}C{1.2cm} | C{1.2cm}C{1.2cm}|}
%\multicolumn{9}{c}{\texttt{galaxy} data} \\
\multicolumn{1}{c}{\texttt{bimod}}
& \multicolumn{2}{c}{DP}
& \multicolumn{2}{c}{PY}
& \multicolumn{2}{c}{ESB}
& \multicolumn{2}{c}{GP} \\ \hline
& $D_v$ & $k_n$ & $D_v$ & $k_n$ & $D_v$ & $k_n$ & $D_v$ & $k_n$ \\ \hline
Mar & 10.51(0.23) & 5.99(0.08) & 27.43(0.88) & 4.67(0.09) &  --- & --- &  --- & ---  \\
OAS & 11.74(0.30) & 9.82(0.18) & 43.44(2.20) & 6.63(0.13) & 34.62(1.11) & 17.42(0.69) & 69.67(2.07) & 57.46(1.68) \\
oOAS & 12.20(0.38) & 12.18(0.22) & 47.95(2.54) & 10.65(0.22) & 43.03(1.78) & 17.60(0.57) & 91.34(2.90) & 78.34(2.44) \\
Cond. & 36.24(1.85) & 66.13(3.80) & 126.6(8.56) & 56.47(3.49) & 66.41(6.13) & 44.65(2.10) & 203.1(10.2) & 161.3(8.00) \\ \hline
\end{tabular}

\vspace{0.5cm}

\begin{tabular}{|c | C{1.2cm}C{1.2cm} | C{1.2cm}C{1.2cm} | C{1.2cm}C{1.2cm} | C{1.3cm}C{1.3cm}|}
%\multicolumn{9}{c}{\texttt{galaxy} data} \\
\multicolumn{1}{c}{\texttt{mix}}
& \multicolumn{2}{c}{DP}
& \multicolumn{2}{c}{PY}
& \multicolumn{2}{c}{ESB}
& \multicolumn{2}{c}{GP} \\ \hline
& $D_v$ & $k_n$ & $D_v$ & $k_n$ & $D_v$ & $k_n$ & $D_v$ & $k_n$ \\ \hline
Mar & 22.32(0.47) & 15.20(0.31) & 22.07(0.52) & 12.88(0.30) &  --- & --- &  --- & ---  \\
OAS & 31.19(0.78) & 24.06(0.59) & 31.52(0.88) & 19.59(0.50) & 49.73(1.72) & 39.47(1.40) & 49.85(1.92) & 35.10(1.35) \\
oOAS & 36.09(1.11) & 30.59(0.81) & 34.09(0.97) & 31.12(0.83) & 55.72(2.39) & 49.27(1.63) & 54.23(2.37) & 39.69(1.43)\\
Cond. & 97.40(5.36) & 96.96(5.19) & 94.09(5.59) & 71.84(3.04) & 255.0(31.8) & 179.1(18.5) & 27212(15981) & 13872(7951) \\ \hline
\end{tabular}

\caption{IATs (standard errors in parenthesis) for the four dataset by model and sampler.\label{tab:IAT}}
\end{scriptsize}
\end{table}

While the IAT reports the efficacy of an algorithm in terms of mixing properties, it does not take into consideration how long it takes to run. The running times were then recorded and combined with the IAT to provide a measurement of overall efficiency of algorithms. Specifically, we computed
\[
\hat{E}_{\mathrm{sampler}} = \frac{\hat{\tau}_{\mathrm{sampler}}t_{\mathrm{sampler}}}{\hat{\tau}_{\mathrm{OAS}}t_{\mathrm{OAS}}}
\]
where $\hat{\tau}_{\mathrm{sampler}}$ and $t_{\mathrm{sampler}}$ stand for the estimated IAT and the average running time per iteration, respectively, of a given sampler. The quantity $\hat{\tau}_{\mathrm{sampler}}/\hat{\tau}_{\mathrm{OAS}}\times N$ can be interpreted as the number of iterations of a given sampler that are needed to obtain the efficiency of estimation of $N$ iterations of the OAS. Hence, $\hat{E}_{\mathrm{sampler}}$ represents 
%quantifies 
the proportion of time, relative to the OAS, that it takes for the chosen sampler to provide an estimate as efficient as the one of the OAS.
In particular, $\hat{E}_{\mathrm{sampler}} < 1$ means that the given sampler was overall more efficient than the OAS, whilst $\hat{E}_{\mathrm{sampler}} > 1$ means that the OAS was more efficient than the given sampler. Table \ref{tab:IAT_R} reports the results. It is important to highlight that the running time might vary depending on the coding of algorithms and the programming language. We used the \texttt{Julia} language and the code is available as supplementary material.

\begin{table}
\begin{scriptsize}
\centering

\begin{tabular}{|c | C{1cm}C{1cm} | C{1cm}C{1cm} | C{1cm}C{1cm} | C{1cm}C{1cm}|}
%\multicolumn{9}{c}{\texttt{galaxy} data} \\
\multicolumn{1}{c}{\texttt{galaxy}}
& \multicolumn{2}{c}{DP}
& \multicolumn{2}{c}{PY}
& \multicolumn{2}{c}{ESB}
& \multicolumn{2}{c}{GP} \\ \hline
& $D_v$ & $k_n$ & $D_v$ & $k_n$ & $D_v$ & $k_n$ & $D_v$ & $k_n$ \\ \hline
Mar & 0.7288 & 0.7108 & 0.7261 & 0.7177 & --- & --- & --- & ---\\
oOAS & 2.2489 & 2.5097 & 1.7932 & 3.8072 & 2.0021 & 3.5688 & 1.8594 & 3.9118 \\
Cond. & 5.6392 & 7.2836 & 39.610 & 60.217 & 9.3107 & 12.604 & 3.8030 & 401.26 \\ \hline
\end{tabular}

\vspace{0.5cm}

\begin{tabular}{|c | C{1cm}C{1cm} | C{1cm}C{1cm} | C{1cm}C{1cm} | C{1cm}C{1cm}|}
%\multicolumn{9}{c}{\texttt{galaxy} data} \\
\multicolumn{1}{c}{\texttt{lepto}}
& \multicolumn{2}{c}{DP}
& \multicolumn{2}{c}{PY}
& \multicolumn{2}{c}{ESB}
& \multicolumn{2}{c}{GP} \\ \hline
& $D_v$ & $k_n$ & $D_v$ & $k_n$ & $D_v$ & $k_n$ & $D_v$ & $k_n$ \\ \hline
Mar & 0.8604 & 0.7016 & 0.8414 & 0.6683 & --- & --- & --- & ---\\
oOAS & 1.6747 & 2.4242 & 1.5401 & 3.0923 & 1.6245 & 2.0384 & 2.1398 & 2.4540 \\
Cond. & 3.5216 & 4.4976 & 6.1502 & 8.8148 & 6.6472 & 3.6405 & 2.9918 & 2.9877 \\ \hline
\end{tabular}

\vspace{0.5cm}

\begin{tabular}{|c | C{1cm}C{1cm} | C{1cm}C{1cm} | C{1cm}C{1cm} | C{1cm}C{1cm}|}
%\multicolumn{9}{c}{\texttt{galaxy} data} \\
\multicolumn{1}{c}{\texttt{bimod}}
& \multicolumn{2}{c}{DP}
& \multicolumn{2}{c}{PY}
& \multicolumn{2}{c}{ESB}
& \multicolumn{2}{c}{GP} \\ \hline
& $D_v$ & $k_n$ & $D_v$ & $k_n$ & $D_v$ & $k_n$ & $D_v$ & $k_n$ \\ \hline
Mar & 0.9283 & 0.6325 & 0.6608 & 0.7371 & --- & --- & --- & ---\\
oOAS & 1.7359 & 2.0719 & 1.8418 & 2.6804 & 1.8488 & 1.5028 & 2.0036 & 2.0836 \\
Cond. & 2.7684 & 6.0395 & 3.3191 & 9.7002 & 1.6078 & 2.1483 & 2.4044 & 2.3154 \\ \hline
\end{tabular}

\vspace{0.5cm}

\begin{tabular}{|c | C{1cm}C{1cm} | C{1cm}C{1cm} | C{1cm}C{1cm} | C{1cm}C{1cm}|}
%\multicolumn{9}{c}{\texttt{galaxy} data} \\
\multicolumn{1}{c}{\texttt{mix}}
& \multicolumn{2}{c}{DP}
& \multicolumn{2}{c}{PY}
& \multicolumn{2}{c}{ESB}
& \multicolumn{2}{c}{GP} \\ \hline
& $D_v$ & $k_n$ & $D_v$ & $k_n$ & $D_v$ & $k_n$ & $D_v$ & $k_n$ \\ \hline
Mar & 0.7772 & 0.6861 & 0.7787 & 0.7312 & --- & --- & --- & ---\\
oOAS & 1.9317 & 2.1226 & 1.8886 & 2.7741 & 1.7980 & 2.0032 & 1.7570 &  1.8263\\
Cond. & 2.8607 & 3.6917 & 10.018 & 12.307 & 3.9647 & 3.5084 & 362.23 & 262.24 \\ \hline
\end{tabular}

\caption{$\hat{E}_{\mathrm{sampler}}$ for the four datasets  by model and sampler.\label{tab:IAT_R}}
\end{scriptsize}
\end{table}

According to Table \ref{tab:IAT}, the most efficient sampler resulted to be the marginal sampler (when applicable), followed by the OAS and then by the oOAS. Taking into consideration both the IAT and the running times, see Table \ref{tab:IAT_R}, one finds that the OAS compares better to the marginal sampler, when applicable, than the oOAS does to the OAS. That is $\hat{E}_{\mathrm{Mar}}^{-1} < \hat{E}_{\mathrm{oOAS}}$. Note also that the OAS is roughly two times more efficient than the oOAS and that the improvement is more significant in terms of $k_n$. This is explained by the deficiency (b) of the oOAS discussed at the end of Section \ref{sec:OAS}, which does not affect the OAS. It is also worth highlighting that if running times are not considered (Table \ref{tab:IAT}), the difference between the performance of the oOAS and the OAS is more subtle. This is due to the fact that each time an ordered allocation variable is updated, the oOAS requires to compute its set of admissible moves, which is not necessary in the OAS; this makes the latter significantly faster than the former to run. As for the conditional sampler, we observe that its worse overall performance is mainly due to its high IAT. In fact, for the DP, ESB and GP models, the dependent slice-efficient sampler was the fastest among all considered samplers. This is not the case of the PY model, for which the conditional sampler was slower than the marginal sampler and the OAS in all four datasets. In particular, for the \texttt{galaxy} and the \texttt{mix} data we found the slice sampler was even slower than the oOAS. It is worth mentioning that there exist other conditional samplers such as the independent slice-efficient sampler \citep{Kall:etal:11} that have better running times than the dependent slice sampler for PY models, however, this modification comes at the cost of a higher IAT.

\section{Split and merge moves}\label{sec:SMmoves}

Split-merge moves were originally proposed by \cite{Gre:Ric:01} for trans dimensional conditional samplers and later adapted by \cite{Jai:Nea:04,Jai:Nea:07} to marginal samplers with a Dirichlet process prior. These are Metropolis-Hastings moves that are particularly useful when two or more mixture components have similar parameters, causing the Gibbs sampler to get stuck in a local mode in terms of the partition of data points.
%that does not correspond to the true clustering of data points. 
Broadly speaking, by proposing splits or merges of blocks of the partition, the Markov chain is able to jump across adjacent high-probability regions instead of having to traverse low-probability regions in between them. 
%\purp{\sout{thus preventing the chain from staying long periods around a local mode.}} 
In this section we exploit the similarities of the new version of OAS with marginal samplers discussed in Section \ref{sec:EOAS} to design split-merge moves that apply to models beyond the Dirichlet or Pitman process priors. We focus first in the case where weights are in size-biased order.

We begin by introducing some notation. Let $\bm{\gamma} = (\bm{d}_{1:n},\bm{\tx}_{1:k_n},\bm{\tp}_{1:k_n})$, be the current state of the chain, thus comprising the ordered allocation variables, and discovered component parameters and weights in order of appearance. 
%Also let $\Pi(\bm{d}_{1:n}) = \{D_1,\ldots,D_{k_n}\}$ be the partition defined by blocks $D_j = \{i \leq n : d_i = j\}$. For each $i,i' \in \{1,\ldots,n\}$ with $i < i'$, we define the neighbourhood of the current state, $\mathcal{N}_{i,i'}(\bm{\gamma})$, as follows: 
Also let $\{D_1,\ldots,D_{k_n}\}$ be the partition corresponding to $\bm{d}_{1:n}$, i.e. $D_d = \{i: d_i = d\}$ for $d \in \{1,\ldots,k_n\}$. 
%When needed, we make explicit this relationship by writing $\{D_1,\ldots,D_{k_n}\}=\Pi(\bm{d}_{1:n})$. 
For any pair of distinct observations, $i$ and $j$, we define a \comillas{neighbourhood} of $\bm{\gamma}$, $\mathcal{N}_{i,j,\bm\gamma}$, to be the set of \textit{split} or \textit{merge} states associated to $\bm\gamma$, according to whether $i$ and $j$ are in the same or different blocks of $\{D_1,\ldots,D_{k_n}\}$. When $i$ and $j$ are in the same block, elements of $\mathcal{N}_{i,j,\bm\gamma}$ are termed \emph{splits} of $\bm{\gamma}$ and are obtained by splitting the block containing $i$ and $j$ into two blocks. In the other case, when $i$ and $j$ belong to different blocks, elements of $\mathcal{N}_{i,j,\bm\gamma}$ are termed \emph{merges} of $\bm{\gamma}$ and are attained by merging the blocks with $i$ and $j$.
%associated to the same or different mixture components, respectively. 
In these neighbourhoods, component parameters corresponding to non affected blocks are constrained to remain the same as in $\bm{\gamma}$. This is made precise by the following definition.

%\red{\it[xx 
%Is it strictly necessary to introduce $\mathcal{N}_{i,j,\bm\gamma}$? Is it because of the extension to the OAS or are we suggesting that Jain and Neal should have introduce it themselves to make the description of the split and merge move more clear? Jain and Neal are so well established, why should we play smart and suggest how they should have explain things better? Can we just stick to their notation and treatment and focus instead to the peculiarities of our setting, i.e. the inclusion of the weights and the preserving somehow of the least element order? If I were a referee, or an Associate Editor, I would not be happy. It is a considerable effort to go through this formalization, even if I were not familiar to Jain and Neal papers, I would read first those papers. And I would be annoyed of such a different treatment if I would not find it strictly necessary 
%xx]}
\begin{definition}\label{def:1}
Let $\bm{\gamma} = (\bm{d}_{1:n},\bm{\tx}_{1:k_n},\bm{\tp}_{1:k_n})$, $i\neq j$, and $S=\{l\neq \{i,j\}:\ d_l=d_i \text{ or }d_l=d_j\}$.
\begin{enumerate}[(a)]
\item
\underline{Split of $\bm{\gamma}$}: If $d_i=d_j$, then
%\purp{$\mathcal{N}_{i,j,\bm\gamma}$ contains \emph{splits} of $\bm\gamma$} and
$\bm{\gamma}^* = (\bm{d}^*_{1:n},\bm{\tx}^*_{1:k_n+1},\bm{\tp}^*_{1:k_n+1}) \in \mathcal{N}_{i,j,\bm\gamma}$ when the partition $\{D^*_1,\ldots,D^*_{k_n+1}\}$ induced by $\bm{d}^*_{1:n}$ can be obtained from $\{D_1,\ldots,D_n\}$ by splitting $D_{d_j} = D_{d_i}$ into two blocks, one  containing $i$  and the other one containing $j$, so that $D^*_{d^*_i} \neq D^*_{d^*_j}$ and for every $l\notin S\cup\{i,j\}$, $\tx_{d_{l}^{*}}^{*} = \tx_{d_{l}}$. 
\item 
\underline{Merge of $\bm{\gamma}$}: If $d_i\neq d_j$, then 
%\purp{$\mathcal{N}_{i,j,\bm\gamma}$ contains \emph{splits} of $\bm\gamma$} and
$\bm{\gamma}^* 
= (\bm{d}^*_{1:n},\bm{\tx}^*_{1:k_n-1},\bm{\tp}^*_{1:k_n-1}) 
\in \mathcal{N}_{i,j,\bm\gamma}$ when the partition $\{D^*_1,\ldots,D^*_{k_n-1}\}$ induced by $\bm{d}^*_{1:n}$ is the result of merging  the blocks $D_{d_i} \neq D_{d_j}$ into a single block $D^*_{d^*_i} = D^*_{d^*_j}$, and for every $l\notin S\cup\{i,j\}$, $\tx_{d_{l}^{*}}^{*} = \tx_{d_{l}}$.
\end{enumerate}
\end{definition}

Note that the values of all the weights and those parameters corresponding to the affected components are unrestricted in the definition above.  Note also that while a split may cause a change in all allocation variables so to preserve the least element order, there is only one way in which the allocation variables can change after a merge. In both cases $\bm{\gamma}^*\in\mathcal{N}_{i,j,\bm\gamma}$ is of a different dimension with respect to $\bm\gamma$. Finally, the set $S$ is introduced for notational convenience so to refer to observations being clustered either with $i$ or $j$.

Given the current state, $\bm{\gamma}$, a split-merge move is a Metropolis-Hastings step restricted to the neighborhood $\mathcal{N}_{i,j,\bm\gamma}$: a new state $\bm{\gamma}^*$ is generated from a proposal distribution $Q(\cdot \mid \bm{\gamma})$ supported on $\mathcal{N}_{i,j,\bm\gamma}$, and later accepted with probability
\begin{equation}\label{eq:A}
A(\bm{\gamma}^*\mid\bm{\gamma}) = \min\l\{1,\frac{\pi(\bm{\gamma}^*\mid \bm{y}_{1:n})Q(\bm{\gamma}\mid \bm{\gamma}^*)}{\pi(\bm{\gamma}\mid \bm{y}_{1:n})Q(\bm{\gamma}^*\mid \bm{\gamma})}\r\}.
\end{equation}
This way, the resulting kernel
\begin{equation}\label{eq:K}
K(\bm{\gamma}^*\mid \bm{\gamma}) = Q(\bm{\gamma}^*\mid \bm{\gamma}) + \delta_{\bm{\gamma}}(\{\bm{\gamma}^*\}) 
\int [1-A(\bm{\gamma}^*\mid\bm{\gamma})]Q(\ddr\bm{\gamma}^*\mid \bm{\gamma}),
\end{equation}
has the posterior $\pi(\bm{\gamma}\mid \bm{y}_{1:n})$ as invariant distribution. 
%Additional requirements on the proposal distribution $Q(\cdot \mid \bm{\gamma})$ are in order for the chain to be irreducible and aperiodic.
Ideally, $Q(\cdot \mid \bm{\gamma})$ should be such that the proposed state, $\bm\gamma^* \in \mathcal{N}_{i,j,\bm{\gamma}}$, has high probability of being accepted when it is reasonable to merge or split the chosen components. 
%These properties will be determined by the choice of the proposal kernel $Q$. 
Following \cite{Jai:Nea:04,Jai:Nea:07}, we will rely on restricted Gibbs sampling scans to define suitable proposals, which are described next.

In our context, by a Gibbs sampling scan restricted to 
$\mathcal{N}_{i,j,\bm\gamma}$ we mean performing a full iteration of the OAS parting from a given state $\bm{\gamma}^{L} \in \mathcal{N}_{i,j,\bm\gamma}$ , called \comillas{launch} state, and restricting the sampled value $\bm{\gamma}^*$ to  $\mathcal{N}_{i,j,\bm\gamma}$.
%, as explained in Algorithm \ref{alg:RGS}. 
Specifically, given $\bm{\gamma}^{L} = (\bm{d}_{1:n}^{L}, \bm{\tx}^{L}_{1:k},\bm{\tp}^{L}_{1:k})$, we sample $\bm{\gamma}^*= (\bm{d}^*_{1:n}, \bm{\tx}^*_{1:k},\bm{\tp}^*_{1:k})$ from the conditional distribution
\begin{multline}\label{eq:Gii}
  G_\pier(\bm{\gamma}^* \mid \bm{\gamma}^L) 
%  =Q(\bm\gamma^*\mid \bm\gamma)
  =\\ 
  G_1(\bm{d}^*_{1:n}\mid 
  \bm{d}^L_{1:n},\bm{\tx}^L_{1:k},\bm{\tp}^L_{1:k},\bm{y}_{1:n})\,
  G_2(\bm{\tx}^*_{1:k}\mid \bm{d}^*_{1:n},\bm{\tx}^L_{1:k},\bm{y}_{1:n})\,
  G_3(\bm{\tp}^*_{1:k}\mid \bm{d}^*_{1:n},\bm{y}_{1:n}),
%  \,\Ind_{\{\bm{\gamma}^* \in \mathcal{N}_{i,j,\bm\gamma}\}}.
\end{multline}
Here $k$ is either $k_n+1$ or $k_n-1$ according to whether $\mathcal{N}_{i,j,\bm\gamma}$ is a set of split or merge states, respectively. 
%\\
%\red{\it[xx notation $G_\pier$ has been changed using the %newcommand \comillas{pier}  xx]}
%\\
The distributions $G_1$, $G_2$ and $G_3$ 
%Loosely speaking these
refer to the distributions from which random terms are sampled in the OAS upon restriction to $\mathcal{N}_{i,j,\bm\gamma}$, and we suggest to sample from them in this specific order. 
For $G_1$, if $i$ and $j$ are in the same mixture component w.r.t. $\bm\gamma$, i.e. $\mathcal{N}_{i,j,\bm\gamma}$ is a set of splits of $\bm\gamma$, then we update the clustering of $l \in S$ via unordered allocation variables as in step 1. of Algorithm \ref{alg:EOAS_PY}, constrained 
%to $\bm{d}^*_{1:n}$ in $\bm\gamma^* \in \mathcal{N}_{i,j,\bm\gamma}$, i.e. 
so that $l $ is either allocated with $i$ or $j$, and then compute the corresponding ordered allocation variables, $\bm{d}^*_{1:n}$. Otherwise, if $i$ and $j$ are in different mixture components w.r.t. $\bm{\gamma}$, then $\mathcal{N}_{i,j,\bm\gamma}$ is a set of merges of $\bm\gamma$, 
and there is a unique value ordered allocation variables can take in $\mathcal{N}_{i,j,\bm\gamma}$, so we 
simply set $\bm{d}^*_{1:n} = \bm{d}^{L}_{1:n}$. 
In view of $G_2$, we need to account for a possibly different labeling of blocks of the partitions in the order of appearance as determined by $\bm{d}^*_{1:n}$ in comparison to those determined by $\bm{d}^L_{1:n}$. These are taken care of by a permutation $(\sigma_1,\ldots,\sigma_k)$ of $(1,\ldots,k)$ and, upon setting $\tx^*_d = \tx^{L}_{\sigma_d}$, we update $\tx^*_d$ for $d\in\{d^*_i,d^*_j\}$ as in the OAS. By definition, the restriction to $\mathcal{N}_{i,j,\bm\gamma}$ does not apply to $G_3$, so the latter coincides with the full conditional of discovered weights in order of appearance. Further details are provided in Algorithm \ref{alg:RGS}. 
%They are related to the distributions from which random terms are sampled in the EOAS.
%The restriction to the set of split or merge states $\mathcal{N}_{i,j,\bm\gamma}$ applies only to $G_1$ and $G_2$.
%The subscript  $G_{\mathcal{N}_{i,j,\bm\gamma}}(\cdot \mid \cdot)$ is meant to specify the support, that is the subset of states, either splits or merges, of the current state $\bm\gamma$ as determined by the pair of observations $i$ and $j$.

%\red{\it[xx we should discuss the split vs the merge cases separately as Jain and Neal (2007) do xx]}
%\red{\it[xx we did not use $G$ before. I have changed the order according to Algorithm 3, to be consistent with equation (13) in Jain and Neal (2007), i.e. left to right  xx]}
%In particular, $G(\bm{d}_{1:n}\mid \cdots)$ refers to the distribution of $\bm{d}_{1:n}$ defined by first updating $\bm{c}_{1:n}$ as described in Algorithm \ref{alg:EOAS_PY} and later compute the unique possible value of ordered allocation variables corresponding to $\bm{c}_{1:n}$. If $g$ and $\nu$ are conjugate, $G(\bm{\tx}_{1:k}\mid \cdots)$ coincides with the full conditional $\pi(\bm{\tx}_{1:k}\mid \cdots) = \prod_{j=1}^{k_n}\pi(\tx_j\mid \cdots)$. As for $G(\bm{\tp}_{1:k_n}\mid \cdots)$, depending on how the weights in order of appearance are updated, it might coincide with the full conditional $\pi(\bm{\tp}_{1:k_n}\mid \cdots)$ as it occurs in the Pitman-Yor model.

\begin{algorithm}[ht]
\caption{Restricted Gibbs sampling scan}\label{alg:RGS}
\begin{small}
Let $\bm{\gamma} = (\bm{d}_{1:n},\bm{\tx}_{1:k_n},\bm{\tp}_{1:k_n})$ be the current state of the chain and $i$ and $j$ distinct observations. Also, let $\mathcal{N}_{i,j,\bm\gamma}$ be defined according to Definition \ref{def:1} and $\bm{\gamma}^L\in \mathcal{N}_{i,j,\bm\gamma}$. 
\begin{enumerate}
\item 
Sample $\bm{d}^*_{1:n}$ from
  $G_1(\bm{d}^*_{1:n}\mid
  \bm{d}^L_{1:n},\bm{\tx}^{L}_{1:k},\bm{\tp}^{L}_{1:k},\bm{y}_{1:n})$.

%  \Ind_{\{\bm\gamma^* \in \mathcal{N}_{i,j,\bm\gamma}\}}$.
%Conditioning on the occupied component parameters and weights $\bm{\tx}^{l}_{1:k}$ and $\bm{\tp}^{l}_{1:k}$, sample $\bm{d}^*_{1:n}$ via  sampling the usual allocation variables $\bm{c}^*_{1:n}$ from their full conditionals as in the EOAS, taking into consideration the restriction that $\bm{\gamma}^* \in \mathcal{N}_{i,i'}(\bm{\gamma})$. 
%Here we recognize two different scenarios.
%\begin{enumerate}[(a)]
%\item
If $d_i=d_j$, set $\bm{c}_{1:n}=\bm{d}^L_{1:n}$, 
%and note that by construction, $c_i\neq c_j$. Update $c_l$ for $l\in S$ as follows
then update $c_l$ for $l\in S$ to be either equal to $c_i$ or $c_j$ according to
%set $j = d^l_i$ and $j' = d^l_{i'}$. The restriction $\bm{\gamma}^*\in \mathcal{N}_{i,i'}(\bm{\gamma})$ means that some allocation variables will not be updated, these are $c_{\iota}^{*}$ such that $c_{\iota}^{*} \not\in \{j,j'\}$ or ${\iota} \in \{i,i'\}$. \red{\it[xx here we are using the trick of the EOAS, right? I mean, moving temporarily to unordered allocation variables xx]} The remaining allocation variables must belong to $\{j,j'\}$, thus they will be sampled from
\begin{equation*}\label{eq:split_c}
  \Prob(c_l =c\mid c_i,c_j,\bm{\tx}^L_{1:k},\bm{\tp}^L_{1:k},y_l) 
  =\frac{\tp^L_{c}\,g(y_k\mid \tx^L_{c})}
  {\tp^{L}_{c_i}\,g(y_k\mid \tx^{L}_{c_i})+\tp^{L}_{c_j}\,g(y_k\mid \tx^{L}_{c_j})} \quad   
  \text{ for } c \in \{c_i,c_j\}.
%  \tp^{L}_c g(y_{\iota}\mid \tx^{L}_c)\Ind_{\{c \in \{j,j'\}\}}, 
\end{equation*}
%This way each $l \in S$ is either allocated with $i$ or $j$. 
%Due to the restriction, there is need to reallocate other indexes $l \not\in S$.}
%\item[(ii)]
%with $\tp^{l}_c = p^{l}_{\alpha^{l}_c}$ when $\bm{\tp}^{l}_{1:k}$ is determined by $(\bm{p}^{l}_{1:\overline{\bm{\alpha}}^{l}_k},\bm{\alpha}^{l}_{1:k})$ . 
Then, for  the distinct values, $\sigma_1,\ldots,\sigma_{k}$, of $\bm{c}_{1:n}$ in order of appearance,
%Afterwards, compute the ordered allocation variables $\bm{d}^*_{1:n}$ by 
set $d^*_i = d$ when $c_i = \sigma_d$, for $i\in \{1,\ldots,n\}$.
%($c_i$ is equal to the $d$th distinct value to appear in $\bm{c}_{1:n}$).
%Thus the distribution from which $\bm{d}^*_{1:n}$ is sampled is
%\begin{equation}\label{eq:split_d}
%\pi_{\mathcal{N}_{i,i'}(\bm{\gamma})}(\bm{d}^*_{1:n}\mid \cdots) = \prod_{l=1}^n \pi_{\mathcal{N}_{i,i'}(\bm{\gamma})}(c_l^*\mid \cdots)\Ind_{\{\bm{d}^*_{1:n} = \mathrm{OAV}(\bm{c}^*_{1:n})\}}
%\end{equation}
%where $\bm{d}^*_{1:n} = \mathrm{OAV}(\bm{c}^*_{1:n})$ means  $\bm{d}^*_{1:n}$  are the ordered allocation variables corresponding to   $\bm{c}^*_{1:n}$, $\pi_{\mathcal{N}_{i,i'}(\bm{\gamma})}(c_l^*\mid \cdots) = \Ind_{\{c^*_l = d^{l}_l\}}$ if $l$ is such that $d^{l}_l \not\in \{j,j'\}$ or $l \in \{i,i'\}$, and $\pi_{\mathcal{N}_{i,i'}(\bm{\gamma})}(c_l^*\mid \cdots)$ is given by \eqref{eq:split_c} otherwise.
%At this step we also permute observed component parameters and weights accordingly. This is we apply the permutation $\bm{\rho} = (\rho_1,\ldots,\rho_{k})$ to the indexes of $\bm{\tx}^{l}_{1:k}$ and $\bm{\tp}^{l}_{1:k}$ obtaining $\bm{\tx}^{*}_{1:k} = (\tx^{l}_{\rho_1},\ldots,\tx^{l}_{\rho_{k}})$ and $\bm{\tp}^{*}_{1:k} = (\tp^{l}_{\rho_1},\ldots,\tp^{l}_{\rho_{k}})$. In particular, when $\tp^{l}_{j} = p^{l}_{\alpha^{l}_j}$ is it enough to apply $\bm{\rho}$ to $\bm{\alpha}^{l}_{1:k}$ obtaining $\bm{\alpha}^{*}_{1:k} = (\alpha^{l}_{\rho_1},\ldots,\alpha^{l}_{\rho_{k}})$ and set $\tp^*_j = p^{l}_{\alpha^*_j}$.
%\end{itemize}
%\item 

If $d_i\neq d_j$,
%If $i$ and $j$ are in different mixture components w.r.t. $\bm{\gamma}$, i.e. $\mathcal{N}_{i,j,\bm\gamma}$ is a set of merges of $\bm\gamma$, 
%then there is a unique value ordered allocation variables can take in $\mathcal{N}_{i,j,\bm\gamma}$ so we 
set $\bm{d}^*_{1:n} = \bm{d}^{L}_{1:n}$ 
%In this case the distribution from which $\bm{d}^*_{1:n}$ is sampled is
%\begin{equation}\label{eq:merge_d}
%\pi_{\mathcal{N}_{i,i'}(\bm{\gamma})}(\bm{d}^*_{1:n}\mid \cdots) = \Ind_{\{\bm{d}^*_{1:n} = \bm{d}'_{1:n}\}}.
%\end{equation}
%and, in view of step 2., set 
and $\sigma_d=d$ for $d \in \{1,\ldots,k\}$.
%\end{enumerate}
%
\item 
Sample $\bm{\tx}^*_{1:k}$ from 
  $G_2(\bm{\tx}^*_{1:k}\mid \bm{d}^*_{1:n},\bm{\tx}^L_{1:k},\bm{y}_{1:n})$.
%  \,\Ind_{\{\bm\gamma^*\in \mathcal{N}_{i,j,\bm\gamma}\}}$.  

%Component parameters of unmodified blocks remain the same. 
%The integers $\rho_1,\ldots,\rho_k$ introduced in step 1. account for a possibly different labeling of blocks of the partitions in the order of appearance as determined by $\bm{d}^*_{1:n}$ in comparison to those determined by $\bm{d}^L_{1:n}$. First 
Set $\tx^*_d = \tx^{L}_{\sigma_d}$ for $d \in \{1,\ldots,k\}$. For $d \in \{d^*_i,d^*_j\}$, sample $\tx^*_d$ from
\begin{equation}\label{eq:pier}
  \pi(\tx^*_d\mid \cdots) \propto 
  \prod_{i:  \in D^*_d}g(y_i\mid \tx^*_d)\nu(\tx^*_d),\quad
  D^*_d = \{i : d^*_i = d\}.
\end{equation}  
%As for $j \in \{d^*_i,d^*_{i'}\}$ we simply sample $\tx^*_j$ as in the OAS. If $g$ and $\nu$ are conjugate this is from its full conditional $\pi(\tx^*_j\mid \cdots) \propto \prod_{{\iota} \in D^*_j}g(y_{\iota}\mid \tx^*_j)\nu(\tx^*_j)$, 
%Notice that if $\bm{\gamma}^*$ represents a merge $d^*_i = d^*_{i'}$, hence only one component parameter will be updated in this case.
%
\item 
Sample $\bm{\tp}^*_{1:k}$ from 
  $G_3(\bm{\tp}^*_{1:k}\mid \bm{d}^*_{1:n},\bm{y}_{1:n})$.  
  
%  \,\Ind_{\{\bm{\gamma}^* \in \mathcal{N}_{i,j,\bm\gamma}\}}$.
%In this case the restriction on the set $\mathcal{N}_{i,j,\bm\gamma}$ does not impose any constraints on the weights, hence these are updated as in the EOAS. Cf. step 3. of Algorithm \ref{alg:EOAS_PY} in the Pitman-Yor case.
Identically as in the OAS, cf. step 3. of Algorithm \ref{alg:EOAS_PY} in the Pitman-Yor case.
%This is, when weights in ordered of appearance are modelled directly we sample $\bm{\tp}^*_{1:k}$ from its full conditional 
%$G_{\mathcal{N}_{i,i'}(\bm{\gamma})}(\bm{\tp}^*_{1:k}\mid \cdots) = G(\bm{\tp}^*_{1:k}\mid \bm{d}^*_{1:n},\bm{y}_{1:n})\Ind_{\{\bm{\gamma}^*\in \mathcal{N}_{i,i'}(\bm{\gamma})\}} = $ 
%$\pi(\bm{\tp}^*_{1:k}\mid \bm{d}^*_{1:n},\bm{y}_{1:n})$. 
%In the other case, we first sample $\bm{\alpha}^*_{1:k}$ conditioning on $\bm{p}^{l}_{1:\overline{\bm{\alpha}}^{l}_k}$ and $\bm{d}^*_{1:n}$, and later we sample $\bm{p}^*_{1:\overline{\bm{\alpha}}^*_k}$ conditioning on $\bm{\alpha}^*_{1:k}$  and $\bm{d}^*_{1:n}$.
\end{enumerate}
Set $\bm{\gamma}^* = (\bm{d}^*_{1:n}, \bm{\tx}^*_{1:k},\bm{\tp}^*_{1:k})$ so that $\bm{\gamma}^*$ is a sample from \eqref{eq:Gii}.
\end{small}
\end{algorithm}

With the aid of restricted Gibbs sampling scans, we define a random proposal distribution, $Q(\,\cdot\mid\bm\gamma)$ as follows. Select two distinct observations, $i$ and $j$, at random uniformly. 
%We refer to $w_{i,i'} > 0$ in general for this probability.
%First choose two indexes $i < i' \in \{1,\ldots,n\}$ with probability $w_{i,i'} > 0$. 
They determine the neighbourhood, $\mathcal{N}_{i,j,\bm\gamma}$, of the current state $\bm\gamma$, so also whether a split or a merge move is proposed. 
%Now, since it is not possible to  transition from $\bm{\gamma}$ to the space $\mathcal{N}_{i,i'}(\bm{\gamma})$ through $G_{\mathcal{N}_{i,i'}(\bm{\gamma})}$ alone \red{\it[xx not clear xx]}, 
Next, we choose a suitable launch state $\bm{\gamma}^L$ as described in Algorithm \ref{alg:launch}. To do so we first initialize $\bm{\gamma}^L \in \mathcal{N}_{i,j,\bm\gamma}$ by either randomly splitting or merging the blocks containing $i$ and $j$, and sampling the component parameters of modified clusters from their prior, as well as all observed weights.
%as detailed in Algorithm \ref{alg:launch}. Then 
Then we update $\bm{\gamma}^L$ by means of $r$ \emph{intermediate} restricted Gibbs sampling scans. The convenience of choosing $\bm{\gamma}^L$ in this way is that, if $r$ is sufficiently large, the last update of $\bm{\gamma}^L$ will be roughly a sample from the restriction of the posterior to $\mathcal{N}_{i,j,\bm\gamma}$. 
Afterwards, we propose $\bm{\gamma}^*$ by performing a final restricted Gibbs sampling scan from the last updated launch state, so that
  $$Q(\bm\gamma^* \mid \bm{\gamma}) 
  =G_\pier(\bm\gamma^* \mid \bm\gamma^L).$$
Notice that the proposed state $\bm{\gamma}^*$ will also be roughly be a sample from the restricted posterior, which increases the probability that the proposed jump to $\mathcal{N}_{i,j,\bm\gamma}$ is accepted in the Metropolis--Hastings algorithm. %A simplification occurs in case of a merge move, since there is a unique value ordered allocation variables can take in  $\mathcal{N}_{i,j,\bm\gamma}$. No restricted Gibbs scans are then necessary for the derivation of the launch state (hence $r=0$) since the final one is sufficient to have $\bm{\gamma}^*$ sampled from the restricted posterior. 
Note that in order to compute the acceptance probability $A(\bm{\gamma}^*\mid \bm{\gamma})$ in \eqref{eq:A}, one needs to compute $Q(\bm\gamma \mid \bm{\gamma}^*)$. The latter is defined analogously to $Q(\bm\gamma^* \mid \bm{\gamma})$:
% Now, in order to compute the acceptance probability $A(\bm{\gamma}^*\mid \bm{\gamma})$ in \eqref{eq:A} of the Metropolis--Hastings procedure, we require a second launch state, say $\bm{\gamma}^{*L}$, corresponding to the proposed state, $\bm{\gamma}^*$, and lying in the neighbourhood $\mathcal{N}_{i,j,\bm\gamma^*}$ of $\gamma^*$ so to define 
  $$Q(\bm\gamma \mid \bm{\gamma}^*) 
  = G_\pier(\bm\gamma \mid \bm\gamma^{*L}),$$
where $\bm{\gamma}^{*L} \in \mathcal{N}_{i,j,\bm\gamma^*}$
is obtained via Algorithm \ref{alg:launch}, replacing $\bm{\gamma}$ and $\bm{\gamma}^L$ with  $\bm{\gamma}^*$ and $\bm{\gamma}^{*L}$, respectively. 
%for a valid definition of the kernel $Q$. 
Notably,  $\bm{\gamma}^{*L}$ is of the same dimension of the current state $\bm{\gamma}$.  
%\purp{ A split-merge move is summarized in Algorithm \ref{alg:SM}.} 
The split-merge procedure finalizes by either accepting the jump from $\bm{\gamma}$ to $\bm{\gamma}^*$ with probability $A(\bm{\gamma}^*\mid \bm{\gamma})$ or rejecting it and letting the chain stay at $\bm{\gamma}$.

%\purp{Before we move on there are a few things to highlight about the computability of the acceptance probability $A(\bm{\gamma}^*\mid \bm{\gamma})$.} 
We discuss next more in details the computation of the acceptance probability. It involves $Q(\bm\gamma^* \mid \bm{\gamma})$ and $Q(\bm\gamma \mid \bm{\gamma}^*)$ together with the ratio 
  $\pi(\bm\gamma^*\mid \bm{y}_{1:n})/\pi(\bm\gamma\mid \bm{y}_{1:n})$.
%Although we can theoretically define a launch state $\hat{\bm{\gamma}}^l$ (via Algorithm \ref{alg:launch}) as well as the distribution $Q(\cdot \mid \hat{\bm{\gamma}}) = G_{\mathcal{N}_{i,i'}(\hat{\bm{\gamma}})}\l(\cdot \mid \hat{\bm{\gamma}}^l\r)$ for any possible configuration $\hat{\bm{\gamma}}$, the distributions $Q(\cdot \mid \bm{\gamma}^*)$ and $Q(\cdot \mid \bm{\gamma}) $ are enough to evaluate the acceptance probability $A(\bm{\gamma}^*\mid \bm{\gamma})$ in \eqref{eq:A}, and proceed with the acceptance-rejection step of the Metropolis--Hastings algorithm. 
The latter does not pose any particular problem, as the normalization constants cancel out, nonetheless care is needed in the evaluation of $Q(\bm\gamma^* \mid \bm{\gamma})$ and $Q(\bm\gamma \mid \bm\gamma^*)$, so in turn of $G_\pier$ in \eqref{eq:Gii}. The distribution $G_1$  is the product of probabilities as detailed in Algorithm \ref{alg:RGS}, step 1., however,
the availability of $G_2$ and $G_3$ in closed form requires further attention. If $g$ and $\nu$ form a conjugate pair, the normalization constant of $G_2$ is easy to obtain. Otherwise
%Nonetheless, they are typically easy to compute and will be given by the product of the normalization constants corresponding to the restricted distributions $G_{\mathcal{N}_{i,i'}(\bm{\gamma})}(\bm{d}^*_{1:n}\mid \cdots)$, 
%$G_{\mathcal{N}_{i,i'}(\bm{\gamma})}(\bm{\tx}^*_{1:k}\mid \cdots)$ and $G_{\mathcal{N}_{i,i'}(\bm{\gamma})}(\bm{\tp}^*_{1:k}\mid \cdots)$ defined in Algorithm \ref{alg:RGS}. Now, the normalization constant of $G_{\mathcal{N}_{i,i'}(\bm{\gamma})}(\bm{d}^*_{1:n}\mid \cdots)$ is always tractable, 
we confine ourselves, as in \cite{Jai:Nea:07}, to the case of the pair $g$ and $\nu$ being conditionally conjugate in one model parameter if the remaining parameters are held fixed. Even though their joint full conditional distribution is not analytically tractable, i.e. we cannot compute explicitly the normalization constant in \eqref{eq:pier}, we can replace the joint density $\pi(\tx^*_d\mid \cdots)$ with the product of conditional densities pertaining subsets of model parameters, see \cite{Jai:Nea:07}.  As for $G_3$, the Pitman-Yor case do not pose any problem: the full conditionals are available in close form, cf. Algorithm \ref{alg:EOAS_PY}.

\begin{algorithm}[ht]
\caption{Launch state $\bm{\gamma}^L$ $\in \mathcal{N}_{i,j,\bm{\gamma}}$}\label{alg:launch}
\begin{small}
Let $\bm{\gamma} = (\bm{d}_{1:n},\bm{\tx}_{1:k_n},\bm{\tp}_{1:k_n})$, be the current state of the chain and $i$ and $j$ distinct observations. 
%We define a launch state $\bm\gamma^L$ corresponding to $\bm{\gamma}$ as follows:
\begin{enumerate}
\item Initialize $\bm{\gamma}^L=(\bm{d}^L_{1:n},\bm{\tx}^L_{1:k},\bm{\tp}^L_{1:k})$ in $\mathcal{N}_{i,j,\bm\gamma}$:
\begin{enumerate}[(a)]
\item
Set $\bm{c}_{1:n}=\bm{d}_{1:n}$. If $i$ and $j$ belong to the same block w.r.t $\bm\gamma$, i.e. $d_i = d_j$, fix $k = k_n+1$, $c_j = k$ and independently for $l\in S$, randomly set $c_l$ to either $c_i$ or $c_j$, with equal probability. This is equivalent to randomly split the elements of $D_{d_i} = D_{d_j}$ into two blocks, one containing $i$ and the other one containing $j$.
%
%\item  
If $i$ and $j$ belong to different components of $\bm\gamma$, set $k = k_n-1$, $c_i=c_j$ and, for every $l\in S$, set $c_l=c_j$ so that $D_{d_i}$ and $D_{d_j}$ are merged into a single block.
%
%\item
Then, for the distinct values, $\sigma_1,\ldots,\sigma_{k}$, of $\bm{c}_{1:n}$ in order of appearance,
%Afterwards, compute the ordered allocation variables $\bm{d}^*_{1:n}$ by 
set $d^{L}_{i} =d$ when $c_i = \sigma_d$, for $i\in \{1,\ldots,n\}$.
\item
Set $\tx^L_d = \tx_{\sigma_d}$ for $d \in \{1,\ldots,k\}$. Then for $d \in \{d_i,d_j\}$, sample $\tx^L_d$ from the prior $\nu$. 
\item
Sample $\bm{\tp}^{L}_{1:k}$ from the prior.
\end{enumerate}
\item Update $\bm{\gamma}^L$ through $r$ restricted Gibbs sampling scans. This is, sample $\bm{\gamma}^* \sim G_\pier(\bm{\gamma}^*\mid \bm{\gamma}^L)$, as described in Algorithm \ref{alg:RGS}, and set $\bm{\gamma}^L = \bm{\gamma}^*$ for $r$ times. 
\end{enumerate}
\end{small}
\end{algorithm}

As explained by \cite{Jai:Nea:04,Jai:Nea:07}, the random proposal distributions do not harm the validity of split-merge moves. To spell this out in our context, first note that the randomness of the proposal distribution, and hence that of the Metropolis-Hastings kernel in \eqref{eq:K}, comes from two sources: the choice of the pair $\{i,j\}$ and the choice of the two launch states $\bm{\gamma}^L,\bm{\gamma}^{*L}$. To make it explicit in the notation, write $K_{i,j,\bm{\gamma}^L,\bm{\gamma}^{*L}}$ instead of $K$ in \eqref{eq:K}. Hence a split-merge move can be equivalently described as sampling from the mixture kernel
\begin{equation}\label{eq:KK}
  \mathcal{K}(\bm{\gamma}^*\mid \bm{\gamma}) = \sum_{i<j}w_{i,j}\int  K_{i,j,\bm{\gamma}^L,\bm{\gamma}^{*L}}(\bm{\gamma}^*\mid \bm{\gamma}) \pi(\mathrm{d}\bm{\gamma}^L,\,\mathrm{d}\bm{\gamma}^{*L}\mid \bm{\gamma},i,j),
\end{equation}
where $w_{i,j} = 1/\binom{n}{2}$ is the probability of choosing the unordered pair $\{i,j\}$ and $\pi(\,\cdot\mid \bm{\gamma},i,j)$ refers to the distribution from which launch states are sampled as determined by Algorithm \ref{alg:launch}. 
This is, given $i$, $j$ and $\bm{\gamma}$ both launch states $\bm{\gamma}^L \in \mathcal{N}_{i,j,\bm{\gamma}}$ and $\bm{\gamma}^{*L} \in \mathcal{N}_{i,j,\bm{\gamma}^*}$ are independently sampled via Algorithm \ref{alg:launch}. A key point here is that in order  to sample the second launch state $\bm{\gamma}^{*L} \in \mathcal{N}_{i,j,\bm\gamma^*}$, we do not require to sample the proposed state, $\bm{\gamma}^*$, beforehand. This follows from the crucial fact that in Algorithm \ref{alg:launch}, $\bm{\gamma}^L$ only depends on $\bm{\gamma}$ through $\mathcal{N}_{i,j,\bm\gamma}$ and that for any two choices $\bm{\gamma}^*_1,\bm{\gamma}^*_2 \in \mathcal{N}_{i,j,\bm\gamma}$ we get $\mathcal{N}_{i,j,\bm{\gamma}^*_1} = \mathcal{N}_{i,j,\bm{\gamma}^*_2}$. In other words, to obtain $\bm{\gamma}^{*L}$ it is enough to identify $\mathcal{N}_{i,j,\bm{\gamma}^*}$ for any arbitrary $\bm{\gamma}^* \in \mathcal{N}_{i,j,\bm{\gamma}}$, randomly initialize $\bm{\gamma}^{*L}$ in this neighbourhood and update its value via restricted Gibbs sampling scans. Thus $\mathcal{K}(\bm{\gamma}^*\mid \bm{\gamma})$ in \eqref{eq:KK} is well defined and describes a split-merge move. Furthermore, under the same argument, but assuming $\bm{\gamma}^*$ is the current state and $\bm{\gamma}$ is the proposed state, we get $\pi(\mathrm{d}\bm{\gamma}^L,\,\mathrm{d}\bm{\gamma}^{*L}\mid \bm{\gamma},i,j) = \pi(\mathrm{d}\bm{\gamma}^L,\,\mathrm{d}\bm{\gamma}^{*L}\mid \bm{\gamma}^*,i,j)$, meaning that both launch states are sampled identically by conditioning on either $\bm{\gamma} \in \mathcal{N}_{i,j,\bm{\gamma}^*}$ or on $\bm{\gamma}^* \in \mathcal{N}_{i,j,\bm{\gamma}}$. This yields $\mathcal{K}(\bm{\gamma}^*\mid \bm{\gamma})\pi(\bm{\gamma}\mid \bm{y}_{1:n}) = \mathcal{K}(\bm{\gamma}\mid \bm{\gamma}^*)\pi(\bm{\gamma}^*\mid \bm{y}_{1:n})$, i.e. $\mathcal{K}$ inherits the invariant distribution of $K_{i,j,\bm{\gamma}^L,\bm{\gamma}^{*L}}(\bm{\gamma}^*\mid \bm{\gamma})$, from which the validity of the split-merge procedure follows.

\subsection{Split-merge moves for non size-biased ordered weights}

To keep split-merge moves widely applicable, we discuss next how to adapt them for mixing priors without tractable size-biased ordered weights $\bm{\tp}$. As done in the OAS, we 
%model $\bm{\tp}$ via $\tp_j = p_{\alpha_j}$  where $\bm{p}$ and $\bm{\alpha}$ are as in Theorem \ref{theo:rep_sss}. This way we
replace $\bm{\tp}_{1:k_n}$  with $(\bm{p}_{1:\overline{\bm{\alpha}}_{k_n}},\bm{\alpha}_{1:k_n})$, where $\overline{\bm{\alpha}}_{k_n} = \max\{\bm{\alpha}_{1:k_n}\}$ and recalling that $\tp_j = p_{\alpha_j}$. Hence, now the current state of the chain, $\bm{\gamma} = (\bm{d}_{1:n},\bm{\tx}_{1:k_n},\bm{p}_{1:\overline{\bm{\alpha}}_{k_n}},\bm{\alpha}_{1:k_n})$, includes ordered allocation variables, component parameters in order of appearance, weights in any order, and the indexes which describe the order in which weights were discovered. 
%Despite this change the neighbourhoods $\mathcal{N}_{i,j,\bm\gamma}$ are defined identically as before because they do not impose any constraint to the weights, and the overall architecture of the algorithm will also be preserved. This said, seemingly the easiest 
Seemingly, the easiest way to adapt split-merge moves is to just modify the updating of weights in restricted Gibbs sampling scans, identically as in the OAS. This is replace $G_\pier(\bm{\gamma}^* \mid \bm\gamma^L)$ in \eqref{eq:Gii} with
\begin{equation}\label{eq:Gii_a}
\begin{aligned}
G_\pier(\bm{\gamma}^* \mid \bm\gamma^L)& = G_1(\bm{d}^*_{1:n}\mid 
  \bm{d}^L_{1:n},\bm{\tx}^L_{1:k},\bm{p}^L_{1:{\overline{\bm{\alpha}}^L_k}},\bm{\alpha}^L_{1:k},\bm{y}_{1:n})\,
  G_2(\bm{\tx}^*_{1:k}\mid \bm{d}^*_{1:n},\bm{\tx}^L_{1:k},\bm{y}_{1:n})\times\\
& \quad \quad \quad \quad G_{3,\alpha}(\bm{\alpha}^*_{1:k}\mid \bm{d}^*_{1:n}, \bm{p}^{L}_{1:\overline{\bm{\alpha}}^{L}_k},\bm{\alpha}^L_{1:k},\bm{y}_{1:n})\,G_{3,p}(\bm{p}^*_{1:\overline{\bm{\alpha}}^*_k}\mid \bm{d}^*_{1:n},\bm{\alpha}^{*}_{1:k},\bm{y}_{1:n}),
\end{aligned}
\end{equation}
where $G_1$ and $G_2$ are as in Algorithm \ref{alg:RGS} (substituting $\tp^L_c$ with $p^L_{\alpha^L_c}$), and $G_{3,\alpha}$ and $G_{3,p}$ stand for the distributions from which $\bm{\alpha}_{1:k}$ and $\bm{p}_{1:\overline{\bm{\alpha}}_{k}}$ are sampled in the OAS (recall that the restriction to $\mathcal{N}_{i,j,\bm{\gamma}}$ does not affect the weights). This substitution can be done in the intermediate Gibbs sampling scans that are used to define launch states $\bm\gamma^L$ (see 2 in Algorithm \ref{alg:launch}). However, care must be taken if we use $G_\pier$ in \eqref{eq:Gii_a}  to define the proposal distributions, i.e. set $Q(\bm{\gamma}^*\mid \bm{\gamma}) = G_\pier(\bm{\gamma}^* \mid \bm\gamma^L)$. This because when $k$ is large $\bm{\alpha}^*_{1:k}$ is sampled via several locally informed Metropolis-Hastings steps cf. Appendix \ref{app:OAS}. This makes it challenging to compute the normalization constant of $G_{3,\alpha}(\bm{\alpha}^*_{1:k}\mid \cdots)$,  which hinders the evaluation of the acceptance probability $A(\bm{\gamma}^*\mid \bm{\gamma})$ in \eqref{eq:A}.
% This jeopardizes both the irreducibility of the produced chain and the implementability of the algorithm. 
%To solve this, we instead 
A solution is to set the proposal distributions $Q(\bm{\gamma}^*\mid \bm{\gamma})$ as a slightly simplified version of $G_\pier$ with computable normalization constant. Namely we can set $Q(\bm{\gamma}^*\mid \bm{\gamma}) = H_\pier(\bm{\gamma}^* \mid \bm\gamma^L)$ with
\begin{equation}\label{eq:Hii_a}
\begin{aligned}
H_\pier(\bm{\gamma}^* \mid \bm\gamma^L)& = \pi(\bm{\alpha}^*_{1:k}\mid \bm{p}^{L})G_1(\bm{d}^*_{1:n}\mid 
  \bm{d}^L_{1:n},\bm{\tx}^L_{1:k},\bm{p}^L_{1:{\overline{\bm{\alpha}}^*_k}},\bm{\alpha}^*_{1:k},\bm{y}_{1:n})\times\\
&  \quad \quad \quad \quad \quad \quad \quad \quad  G_2(\bm{\tx}^*_{1:k}\mid \bm{d}^*_{1:n},\bm{\tx}^L_{1:k},\bm{y}_{1:n})
 G_{3,p}(\bm{p}^*_{1:\overline{\bm{\alpha}}^*_k}\mid \bm{d}^*_{1:n},\bm{\alpha}^{*}_{1:k},\bm{y}_{1:n}),
\end{aligned}
\end{equation}
In contrast to $G_\pier$, $H_\pier$ first samples $\bm{\alpha}^*_{1:k}$ as a priori from $\pi(\bm{\alpha}^*_{1:k}\mid \bm{p}^{L})$, at this stage the required extra weights $p^L_j$ (if any)  with $\overline{\bm{\alpha}}^{L}_k < j \leq \overline{\bm{\alpha}}^*_k$ can be sampled sequentially from the conditional distributions, $\pi(p^L_{j+1}\mid \bm{p}^L_{1:j})$, corresponding to the prior law $\pi(\bm{p})$. Since weights corresponding to unobserved components would have been updated in this way, cf. \eqref{eq:post_p}, we can simply assume that the extra weights we required were already included in $\bm{\gamma}^L$. Once $\bm{\alpha}^*_{1:k}$ has been obtained, $H_\pier$ then samples $\bm{d}^*_{1:n}$ from $G_1$ as in Algorithm \ref{alg:RGS} with $\tp^L_{c} =  p^L_{\alpha^*_c}$, $\bm{\tx}^*_{1:k}$ from $G_2$ as in Algorithm \ref{alg:RGS} and $\bm{p}^*_{1:\overline{\bm{\alpha}}^*_k}$ from $G_{3,p}$ as in the OAS. The advantage of using $H_\pier$ instead of $G_\pier$ to define $Q$ is that the normalization constant of $\pi(\bm{\alpha}^*_{1:k}\mid \bm{p}^{L})$ is available, cf. \eqref{eq:size-biased_pick}.
%and that the support of $\pi(\bm{\alpha}^*_{1:k}\mid \bm{p}^{L})$ includes all possible values of $\bm{\alpha}^*_{1:k}$ 
%\red{\it[xx is this relevant only if one were to use split merge moves as a stand alone sampler? xx]}. 
The disadvantage is that the posterior distribution is no longer invariant with respect to $H_\pier$. In spite of this, the probability of accepting proposed states $\bm{\gamma}^* \sim Q(\bm{\gamma}^*\mid \bm{\gamma}) = H_\pier(\bm{\gamma}^* \mid \bm\gamma^L)$ should not be affected significantly. 
%\red{\it[xx better here xx]}}. 
Split-merge moves for non-size-biased ordered weights are fully specified in Algorithm \ref{alg:SM_nsb}. In particular, we detail the choice of the launch states and  the restricted Gibbs sampling scans by comparison to Algorithms \ref{alg:RGS} and \ref{alg:launch}.
%\red{\it[xx better here xx]}

\begin{remark}
An alternative to the choice $Q(\bm{\gamma}^*\mid \bm{\gamma}) = H_\pier(\bm{\gamma}^* \mid \bm\gamma^L)$ is to take $Q(\bm{\gamma}^*\mid \bm{\gamma})$ similar to a restricted Gibbs sampling scan with the difference that, when we updating $\bm{\alpha}^*_{1:k}$, only one locally informed Metropolis-Hastings step is performed, cf. \eqref{eq:pi_rho}. In this case, the posterior remains invariant w.r.t. $Q$ and we can also evaluate its normalization constant. This solution is feasible if split-merge moves are included as an acceleration step in the OAS. However, if split-merges moves are used on their own as a standalone MCMC scheme, a single locally informed Metropolis-Hastings step has reduced support, which might affect the irreducibility of the chain.
\end{remark}

\begin{algorithm}[tb]
\caption{Split-merge moves for non-size-biased ordered weights}\label{alg:SM_nsb}
\begin{small}
Let $\bm{\gamma} = (\bm{d}_{1:n},\bm{\tx}_{1:k_n},\bm{p}_{1:\overline{\bm{\alpha}}_{k_n}},\bm{\alpha}_{1:k_n})$, be the current configuration of the ordered allocation variables, and discovered component parameters and weights. %In particular, when weights in order of appearance are given by $\tp_j = p_{\alpha_j}$, we replace $\bm{\tp}_{1:k_n}$  with the collection $(\bm{p}_{1:\overline{\bm{\alpha}}_{k_n}},\bm{\alpha}_{1:k_n})$
\begin{enumerate}
\item Select two distinct observations, $i$ and $j$ at random uniformly.
\item Sample launch state $\bm\gamma^L \in \mathcal{N}_{i,j,\bm{\gamma}}$:
\begin{itemize}
\item[(i)] First initialize $\bm{d}^{L}_{1:n}$ and $\bm{\tx}^{L}_{1:k}$ as in Algorithm \ref{alg:launch}. Then sample $\bm{\alpha}^{L}_{1:k}$ and $\bm{p}^{L}_{1:\overline{\bm{\alpha}}^{L}_k}$ from their prior.
\item[(ii)] Update $\bm\gamma^L$ through $r$ restricted Gibbs sampling scans, where each time sample $\bm{\gamma}^* \sim G_\pier(\bm{\gamma}^* \mid \bm\gamma^L)$ in \eqref{eq:Gii_a} and set $\bm\gamma^L = \bm{\gamma}^*$. Follow Algorithm \ref{alg:RGS} with the following changes. In step 1. set $\tp^{L}_c = p^{L}_{\alpha^{L}_c}$, then after updating $\bm{d}^*_{1:n}$, set $\bm{\alpha}^L_{1:k} = (\alpha^{L}_{\sigma_1},\ldots,\alpha^{L}_{\sigma_k})$, for $\sigma_1,\ldots,\sigma_k$ defined there. In step 3., sample $\bm{\alpha}^*_{1:k}$ and $\bm{p}^*_{1:\overline{\bm{\alpha}}^*_{k}}$ identically as in the OAS, cf. step 3. in Algorithm \ref{alg:EOAS}, conditioning on $\bm{d}^*_{1:n}$.
\end{itemize}
\item Sample $\bm{\gamma}^* \sim Q(\,\cdot\mid \bm{\gamma}) = H_\pier\l(\cdot \mid \bm\gamma^L\r)$ as in \eqref{eq:Hii_a}: First sample $\bm{\alpha}^*_{1:k}$ from $\pi(\bm{\alpha}^*_{1:k}\mid \bm{p}^{L})$ as a priori along with the require extra weights  $p^L_j$, for $\overline{\bm{\alpha}}^{L}_k < j \leq \overline{\bm{\alpha}}^*_k$, sequentially from $\pi(p^L_{j+1}\mid \bm{p}^L_{1:j})$, when needed. Conditioning on $\bm{\alpha}^*_{1:k}$, sample $\bm{d}^*_{1:n}$ as in Algorithm \ref{alg:RGS} step 1. with $\tp^{L}_c = p^{L}_{\alpha^*_c}$, then set   $\bm{\alpha}^*_{1:k} = (\alpha^*_{\sigma_1},\ldots,\alpha^*_{\sigma_k})$ for $\sigma_1,\ldots,\sigma_k$ defined there. Finally, sample $\bm{\tx}^*_{1:k}$ as in  Algorithm \ref{alg:RGS} step 2., and the weights $\bm{p}^*_{1:\overline{\bm{\alpha}}^*_{k}}$ identically as in the OAS, cf. step 3. in Algorithm \ref{alg:EOAS},  conditioning on $\bm{\alpha}^*_{1:k}$ and $\bm{d}^*_{1:n}$.
\item Sample launch state $\bm{\gamma}^{*L} \in \mathcal{N}_{i,j,\bm{\gamma}^*}$, similarly to 2.(i) and 2.(ii) above, define $Q(\bm\gamma\mid \bm{\gamma}^*) = H_\pier\l(\bm\gamma \mid \bm{\gamma}^{*L} \r)$ and compute $A(\bm{\gamma}^*\mid \bm{\gamma})$ in \eqref{eq:A}.
\item  Accept the move and set $\bm{\gamma} = \bm{\gamma}^*$ with probability $A(\bm{\gamma}^*\mid \bm{\gamma})$, otherwise leave $\bm{\gamma}$ unchanged.
\end{enumerate}
\end{small}
\end{algorithm}

\subsection{Illustration}\label{sec:SM_illust}

For illustration purposes we designed a simulation study that tests whether the inclusion of split-merge moves aids the OAS escape local modes. Roughly speaking, a single run of the experiment consists in generating $100$ data points from a mixture of two or more
%either two or three 
Normal distributions. We then estimate the density of the data and compare it to the true density, by taking into consideration $100$ iterations of the OAS, after two different types of burn-in period. 
%The difference between the two burn-in periods relies on the number of iterations and whether split-merge moves were included or not. 
In the first case we consider $110$ burn-in iterations without split-merge moves, as for the second type we keep the burn-in period as low as $10$ iterations, each iteration including a split-merge move. For the split-merge moves we performed $r = 10$ restricted Gibbs sampling scans. In all cases the $100$ iterations of the OAS, after the burn-in period, do not include split-merge moves. 
We have forced the samplers to start in a local mode by initializing all ordered allocations variables equal to one. This way, all data points are allocated to a single component, which does not correspond to the true clustering.
%In any case, we initialize all ordered allocations variables equal to one, so that all data points are allocated to a single component. This way, the sampler is forced to start in a local mode that does not correspond to the true clustering of data points. %For the OAS without split-merge moves we have considered a burn-in period of $110$ iterations, whilst for the OAS with split-merge moves we keep the burn-in period as low as $10$ iterations, each iteration including a split-merge move with $r = 10$ restricted Gibbs sampling scans. 
Now, to compare the estimated density, $\hat{f}$, 
%both with and without split-merge moves, afterwards we compare $\hat{f}$
with the true density, $f$,
%graphically and 
we computed the total variation distance $d_{\mathrm{TV}}(f,\hat{f}) = \frac{1}{2}\int_{-\infty}^{\infty} |f(x)-\hat{f}(x)| dx$. When the sampler effectively escapes the local mode in the given burn-in period, we expect to observe a small value of $d_{\mathrm{TV}}(f,\hat{f})$. 
The rationale for using a small number of iterations is to pin-point the ability of the sampler to escape a local mode. Indeed, after a  large enough number of iterations, the OAS is expected to escape the local mode even without split-merge moves. In order to avoid bias in the results caused by the small number of iterations coupled with the simulated data being relatively hard to cluster correctly, we have repeated the experiment $100$ times, each time with a different dataset sampled from the given data distribution.
For the true density, $f$, we have considered the \texttt{trimodal} mixture, $0.25\mathsf{N}(-1.4,0.3^2) + 0.5\mathsf{N}(0,0.3^2) + 0.25\mathsf{N}(1.4,0.3^2)$, and the \texttt{bimodal} one, $0.5\mathsf{N}(-1,0.6^2) + 0.5\mathsf{N}(1,0.6^2)$. As for the mixing priors, we chose a Dirichlet process prior which enjoys tractable size-biased ordered weights, as well as a Geometric prior which does not. 
%In Table \ref{tab:SM} we report the results and in Figure \ref{fig:SM} we illustrate one run for each choice of $f$.

\begin{table}
\begin{small}
\centering
\begin{tabular}{|c | C{2cm}C{2cm} | C{2cm}C{2cm}|}
%\multicolumn{9}{c}{\texttt{galaxy} data} \\
\multicolumn{1}{c}{OAS}
& \multicolumn{2}{c}{\texttt{trimodal}}
& \multicolumn{2}{c}{\texttt{bimodal}}\\ \hline
& DP & GP & DP & GP  \\ \hline
without split-merge & 0.2637 (0.037) & 0.2463 (0.063) & 0.1077 (0.053) & 0.1228 (0.033)  \\
with split-merge & 0.1619 (0.085) & 0.2241 (0.064) & 0.0543 (0.026) & 0.0989 (0.035) \\ \hline
\end{tabular}
\caption{Mean (standard errors in parenthesis) of $d_{\mathrm{TV}}(f,\hat{f})$ for the $100$ runs, by sampler (burn-in period considered), choice of $f$ and mixing prior.}
\label{tab:SM}
\end{small}
\end{table}

\begin{figure}
\centering
\includegraphics[scale=0.4]{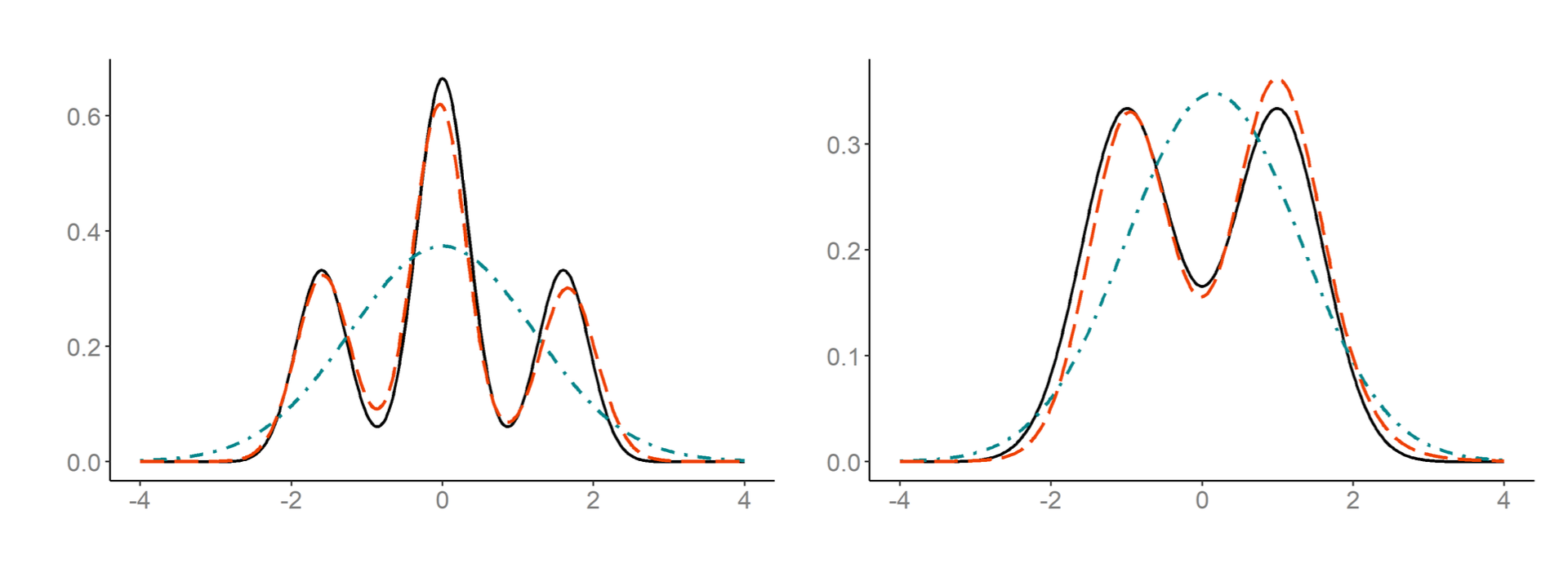}
\caption{Plot of $f$ (solid line) and $\hat{f}$ via the OAS with (dashed line) and without (dot-dashed line) split-merge moves. The graph on the left corresponds to the \texttt{trimodal} mixture using a DP prior and the one on the right pertains the \texttt{bimodal} mixture using a GP prior.}\label{fig:SM}
\end{figure}

As it can be observed in Table \ref{tab:SM}, $d_{\mathrm{TV}}(f,\hat{f})$ is on average smaller when split-merge moves are included in the burn-in period. Thus confirming that split-merge moves help the OAS escape the local mode in only $10$ iterations, while $110+100$ iterations without split-merge moves are not sufficient more often than not. To have a visual representation of the experiment, in Figure \ref{fig:SM} we have plotted the true (solid line) and the two estimated densities (dashed lines and dot-dashed lines correspond to the samplers with and without split-merge moves, respectively) for a single run of the experiment. The graph in the left shows the results for the \texttt{trimodal} mixture with a DP prior and the one on the right corresponds to the \texttt{bimodal} mixture with a GP prior. In each case, the run was chosen because the simulated dataset was representative of the true density, making it relatively easy to cluster data points correctly. In the DP case, total variation distances are $0.2913$ and $0.0519$ without and with split-merge moves, respectively, and in the GP case, they are $0.1661$ and $0.0414$  without and with split-merge moves, respectively. Notice that in comparison to Table \ref{tab:SM} the total variation distances with split-merge moves are small. This does not mean, however, that in most runs the OAS with split-merge moves struggled to escape the local mode. It simply reflects the fact that not all datasets are as representative of the true distribution, in which case $d_{\mathrm{TV}}(f,\hat{f})$ might not be so close to zero  even when the sampler has escaped the local mode.

\section{Final remarks}

In this paper we have derived an enhanced version of the order allocation sampler (OAS) of \cite{Deb:Gil:23} that runs fast, mixes remarkably well when compared to marginal samplers and improves substantially over out-of-the-shelf conditional samplers even when the mixing prior does not enjoy a tractable predictive distribution that allows for marginalization. The key modification is that ordered allocation variables can be updated sequentially as unordered ones, so the sampler evolves in the space of partitions like marginal samplers, as opposed to the space of ordered partitions with blocks in the least element order as in the original version of the OAS.
%a change of space when updating the ordered allocation variables. Namely, the original version of the OAS evolves in the space of ordered partitions with blocks in the least element order. Instead this new variant transfers to the space of unordered partitions when updating allocation variables, and later returns to the space of ordered partitions so to proceed with the rest of the sampler identically as the original OAS. 
To this aim, we leverage on the one-to-one correspondence between unordered and ordered partitions with blocks in the least element order, together with the exchangeability of the data and key distributional properties of the weights in size-biased order. Exploiting the similarities with marginal samplers, we are then able to adapt the split-merge moves of \cite{Jai:Nea:04,Jai:Nea:07} to the OAS. This adaptation would have been rather challenging with the original OAS because of the constraint induced by the least element ordered, cf. $\Ind_{\D}$ in \eqref{eq:like_sb}, when splitting or merging  blocks of the partition.  

%The new version of the OAS can in fact be regarded as a conditional version of a marginal sampler. A couple of remarks are in order. 
We conclude with a couple of remarks about the relationship between the new version of the OAS and the marginal sampler.
Similar to Algorithm 8 by \cite{Nea:00}, in the updating allocation variables we generate a potentially new discovered component parameter from the prior, $\tx^*_{k+1} \sim \nu$. It is not difficult to envision here a different updating step that mimics Algorithm 2 by \cite{Nea:00}. Instead of  sampling $c_i$ from \eqref{eq:c_post0} by conditioning on $\tx^*_{k+1}$, one can integrate out this quantity and sample $c_i$ from
\begin{equation}\label{eq:c_post1}
  \Prob(\,c_i = c\mid \rest\,) \propto \begin{cases}
p^*_c\, g(y_i\mid x^*_c) & \text{ if } c \in \{1,\ldots,k\},\\
\l(1-\sum_{l=1}^{k}p^*_l\r) \int g(y_i\mid x^*_c)\nu(\ddr x^*_c)& \text{ if } c = k+1.
\end{cases}
\end{equation}
provided that the integral in \eqref{eq:c_post1} is available in closed form, which is the case of $g$ and $\nu$ forming a conjugate pair. In case $c_i = k+1$, we then sample the associated component parameter $\tx^*_{k+1}$ from $\pi(\tx^*_{k+1} \mid \cdots) \propto g(y_i\mid \tx^*_{k+1})\nu(\tx^*_{k+1})$ as well as the corresponding weight $\tp^*_{k+1}$ as explained in Section \ref{sec:EOAS}. A second remark is to highlight the key differences between the split-merge moves for marginal samplers by \cite{Jai:Nea:04,Jai:Nea:07} and those for the OAS in Section \ref{sec:SMmoves}. Our proposal works on an expanded space due to the inclusion of weights in order of appearance. This means that each restricted Gibbs sampling scan requires to update more random variables and that the computation of the acceptance probability in \eqref{eq:A} does not simplifies as nicely. On the positive side, by conditioning on the weights we can update (usual) allocation variables independently of each other (see 1 in Algorithm \ref{alg:RGS}) because components are neither created nor destroyed in restricted Gibbs sampling scans. This means that this step is parallelizable which is quite important in big-data settings. For marginal split-merge moves it is also true that components are not created or destroyed, however, each time an allocation variables is updated, the quantities $m_c =|\{l: c_l = c\}|$ change, cf.  \eqref{eq:c_post0_mar}. This prohibits updating allocation variables independently of each other in restricted Gibbs sampling scans, and thus the parallelization of this step. Another advantage is that having adapted split-merge moves to the OAS, makes them applicable to a wider range of mixing priors, even to the ones that do not admit marginalization because of the lack of tractable P\'olya urn scheme representation of the predictive distribution.
%\red{\it[xx better here xx]}

%--------------------------------------------------------
%--------------------------------------------------------
%--------------------------------------------------------

\section*{Acknowledgements}

M. F. Gil-Leyva and F. Selva gratefully acknowledge the support of PAPIIT Grant IA101124. F. Selva was also supported by a CONAHCyT PhD scholarship. P. De Blasi acknowledges support of MUR - Prin 2022 - Grant no. 2022CLTYP4, funded by the European Union – Next Generation EU.

%--------------------------------------------------------
%--------------------------------------------------------
%--------------------------------------------------------

%\newpage

%\bibliographystyle{dcu}
\bibliographystyle{apalike}
\bibliography{short_biblio}

\begin{thebibliography}{}

\bibitem[Blackwell and MacQueen, 1973]{Bla:Mac:73}
Blackwell, D. and MacQueen, J.~B. (1973).
\newblock {Ferguson distributions via Polya urn schemes}.
\newblock {\em {Ann. Statist.}}, 1(2):353 -- 355.

\bibitem[Blei and Jordan, 2006]{Blei:Jor:06}
Blei, D.~M. and Jordan, M.~I. (2006).
\newblock {Variational inference for Dirichlet process mixtures}.
\newblock {\em Bayesian Analysis}, 1(1):121 -- 143.

\bibitem[De~Blasi and Gil–Leyva, 2023]{Deb:Gil:23}
De~Blasi, P. and Gil–Leyva, M.~F. (2023).
\newblock Gibbs sampling for mixtures in order of appearance: The ordered
  allocation sampler.
\newblock {\em Journal of Computational and Graphical Statistics}, 0(0):1--9.

\bibitem[Escobar, 1994]{Esc:94}
Escobar, M.~D. (1994).
\newblock Estimating normal means with a dirichlet process prior.
\newblock {\em Journal of the American Statistical Association},
  89(425):268--277.

\bibitem[Escobar and West, 1995]{Esc:Wes:95}
Escobar, M.~D. and West, M. (1995).
\newblock {B}ayesian density estimation and inference using mixtures.
\newblock {\em {J. Amer. Statist. Assoc.}}, 90:577--588.

\bibitem[Favaro et~al., 2016]{Fav:etal:16}
Favaro, S., Lijoi, A., Nava, C., Nipoti, B., Pr\"unster, I., and Teh, Y.~W.
  (2016).
\newblock {O}n the stick-breaking representation for homogeneous {NRMI}s.
\newblock {\em Bayesian Anal.}, 11:697--724.

\bibitem[Favaro et~al., 2012]{Fav:etal:12}
Favaro, S., Lijoi, A., and Pr\"{u}nster (2012).
\newblock On the stick-breaking representation of normalized inverse {G}aussian
  priors.
\newblock {\em Biometrika}, 99:663--674.

\bibitem[Ferguson, 1973]{Fer:73}
Ferguson, T. (1973).
\newblock A {B}ayesian analysis of some nonparametric problems.
\newblock {\em {Ann. Statist.}}, 1(2):209--230.

\bibitem[Fuentes-García et~al., 2010]{Fue:etal:10}
Fuentes-García, R., Mena, R.~H., and Walker, S.~G. (2010).
\newblock A new {B}ayesian nonparametric mixture model.
\newblock {\em Communications in Statistics - Simulation and Computation},
  39(4):669--682.

\bibitem[Gil–Leyva and Mena, 2021]{Gil:Men:21}
Gil–Leyva, M.~F. and Mena, R.~H. (2021).
\newblock {S}tick-breaking processes with exchangeable length variables.
\newblock {\em Journal of the American Statistical Association}, page in press.

\bibitem[Green and Richardson, 2001]{Gre:Ric:01}
Green, P.~J. and Richardson, S. (2001).
\newblock {Modeling heterogeneity with and without the Dirichlet process}.
\newblock {\em Scand. J. Stat.}, 28:355--375.

\bibitem[Ishwaran and James, 2001]{Ish:Jam:01}
Ishwaran, H. and James, L.~F. (2001).
\newblock Gibbs sampling methods for stick-breaking priors.
\newblock {\em J. Amer. Statist. Assoc.}, 96:161--173.

\bibitem[Jain and Neal, 2004]{Jai:Nea:04}
Jain, S. and Neal, R.~M. (2004).
\newblock {A split-merge Markov Chain Monte Carlo procedure for the Dirichlet
  process mixture model}.
\newblock {\em {J. Comput. Graph. Statist.}}, 13(1):158--182.

\bibitem[Jain and Neal, 2007]{Jai:Nea:07}
Jain, S. and Neal, R.~M. (2007).
\newblock {Splitting and merging components of a nonconjugate Dirichlet process
  mixture model}.
\newblock {\em Bayesian Anal.}, 2(3):445 -- 472.

\bibitem[Kalli et~al., 2011]{Kall:etal:11}
Kalli, M., Griffin, J.~E., and Walker, S. (2011).
\newblock Slice sampling mixtures models.
\newblock {\em Statist. Comput.}, 21:93--105.

\bibitem[Lo, 1984]{Lo:84}
Lo, A.~Y. (1984).
\newblock On a class of bayesian nonparametric estimates: I. density estimates.
\newblock {\em {Ann. Statist.}}, 12(1):351--357.

\bibitem[Miller and Harrison, 2018]{Mil:Har:18}
Miller, J.~W. and Harrison, M.~T. (2018).
\newblock Mixture models with a prior on the number of components.
\newblock {\em {J. Amer. Statist. Assoc.}}, 113(521):340--356.

\bibitem[Neal, 2000]{Nea:00}
Neal, R.~M. (2000).
\newblock {Markov Chain Sampling Methods for {D}irichlet Process Mixture
  Models}.
\newblock {\em {J. Comput. Graph. Statist.}}, 9(2):249--265.

\bibitem[Papaspiliopoulos and Roberts, 2008]{Pap:Rob:08}
Papaspiliopoulos, O. and Roberts, G.~O. (2008).
\newblock {Retrospective Markov chain Monte Carlo methods for Dirichlet process
  hierarchical models}.
\newblock {\em Biometrika}, 95:169--186.

\bibitem[Pitman, 1995]{Pit:95}
Pitman, J. (1995).
\newblock Exchangeable and partially exchangeable random partitions.
\newblock {\em Probab. {T}heory {R}elat. {F}ields}, 102:145--158.

\bibitem[Pitman, 1996a]{Pit:96}
Pitman, J. (1996a).
\newblock Random discrete distributions invariant under size-biased
  permutation.
\newblock {\em Adv. {A}ppl. {P}robab.}, 28(2):525--539.

\bibitem[Pitman, 1996b]{Pit:96ims}
Pitman, J. (1996b).
\newblock {Some developments of the Blackwell-MacQueen urn scheme}.
\newblock In et~al., T.~F., editor, {\em Statistics, Probability and Game
  Theory; Papers in honor of {D}avid {B}lackwell}, volume~30 of {\em Lecture
  Notes-Monograph Series}, pages 245--267, Hayward, California. Institute of
  Mathematical Statistics.

\bibitem[Pitman, 2006]{Pit:06}
Pitman, J. (2006).
\newblock {\em {Combinatorial Stochastic Processes}}, volume 1875 of {\em
  \'{E}cole d'\'et\'e de probabilit\'es de {S}aint-{F}lour}.
\newblock Springer-{V}erlag {B}erlin {H}eidelberg, New York, first edition.

\bibitem[Pitman and Yor, 1992]{Pit:Yor:92}
Pitman, J. and Yor, M. (1992).
\newblock Arcsine laws and interval partitions derived from a stable
  subordinator.
\newblock {\em Proceedings of the {L}ondon {M}athematical {S}ociety},
  s3-65(2):326--356.

\bibitem[Pitman and Yor, 1997]{Pit:Yor:97}
Pitman, J. and Yor, M. (1997).
\newblock The two-parameter {P}oisson-{D}irichlet distribution derived from a
  stable subordinator.
\newblock {\em {Ann. Probab.}}, 25(2):855--900.

\bibitem[Porteous et~al., 2006]{Port:etal:06}
Porteous, I., Ihler, A., Smyth, P., and Welling, M. (2006).
\newblock Gibbs sampling for (coupled) infinite mixture models in the stick
  breaking representation.
\newblock In {\em Proceedings of the Twenty-Second Conference on Uncertainty in
  Artificial Intelligence (UAI2006)}, pages 385--392.

\bibitem[Regazzini et~al., 2003]{Reg:etal:03}
Regazzini, E., Lijoi, A., and Pr{\"{u}}nster, I. (2003).
\newblock Distributional results for means of normalized random measures with
  independent increments.
\newblock {\em {Ann. Statist.}}, 31(2):560--585.

\bibitem[Richardson and Green, 1997]{Ric:Gre:97}
Richardson, S. and Green, P.~J. (1997).
\newblock {On Bayesian analysis of mixtures with an unknown number of
  components}.
\newblock {\em {J. R. Stat. Soc. Ser. B}}, 59:731--792.

\bibitem[Sethuraman, 1994]{Seth:94}
Sethuraman, J. (1994).
\newblock {A constructive definition of Dirichlet priors}.
\newblock {\em Stat. Sin.}, 4:639--650.

\bibitem[Sokal, 1997]{Sok:97}
Sokal, A. (1997).
\newblock {Monte Carlo Methods in Statistical Mechanics: Foundations and New
  Algorithms}.
\newblock In DeWitt-Morette, C., Cartier, P., and Folacci, A., editors, {\em
  Functional Integration: Basics and Applications}, pages 131--192, Boston, MA.
  Springer US.

\bibitem[Walker, 2007]{Walk:07}
Walker, S.~G. (2007).
\newblock Sampling the {D}irichlet mixture model with slices.
\newblock {\em Communications in {S}tatistics-{S}imulation and {C}omputation},
  36(1):45--54.

\bibitem[Zanella, 2020]{Zan:20}
Zanella, G. (2020).
\newblock Informed proposals for local {MCMC} in discrete spaces.
\newblock {\em {J. Amer. Statist. Assoc.}}, 115(530):852--865.

\end{thebibliography}

%--------------------------------------------------------
%--------------------------------------------------------
%--------------------------------------------------------

\newpage
\setcounter{figure}{0}
\setcounter{table}{0}
\setcounter{equation}{0}
\renewcommand{\thefigure}{\thesection\arabic{figure}}
\renewcommand{\thetable}{\thesection\arabic{table}}
\renewcommand{\theequation}{\thesection\arabic{equation}}
\renewcommand{\thelemma}{\thesection.\arabic{lemma}}
\renewcommand{\thetheorem}{\thesection.\arabic{theorem}}
\renewcommand{\theremark}{\thesection.\arabic{remark}}

\appendix

\begin{center}
\begin{Large}
{\bf Appendix}
\end{Large}
\end{center}

\section{Derivation of the ordered allocation sampler}\label{app:OAS}

\subsection*{Updating component parameters, $\bm{\tx}$:}

Under the convention that the empty product equals one, from \eqref{eq:like_sb} and \eqref{eq:prior_sb}, or \eqref{eq:like} and \eqref{eq:prior}, we find
\begin{equation}\label{eq:x_post}
\pi(\,\bm{\tx}\mid \cdots) = \prod_{j\geq 1}\pi(\,\tx_j\mid\rest) \propto \prod_{j\geq 1}\prod_{i \in D_j} g(y_i\mid \tx_j) \times \nu(\tx_j).
\end{equation}
This means that the components parameters, $\tx_j$, will be updated independently of each other. For $j > k_n$ we sample $\tx_j \sim \nu$. As for $j \leq k_n$, it is easy to sample from \eqref{eq:x_post} if $g$ and $\nu$ form a conjugate pair. Otherwise, the updating of observed component parameters is treated identically as in other samplers.

\subsection*{Updating weights in order of appearance, $\bm{\tp}$:}

It follows from \eqref{eq:like_sb} and \eqref{eq:prior_sb} that the full conditional  of $\bm{\tp}$ is
\begin{equation}\label{eq:p_post}
\pi(\,\bm{\tp}\mid\rest) \propto \prod_{j=1}^{k_n}\tp_j^{n_j-1}
  \bigg(1-\sum_{l=1}^{j-1}\tp_l\bigg) \times \pi(\bm{\tp}).
\end{equation}
For some models such as the Dirichlet, the Pitman-Yor and NRMI models \citep{Ish:Jam:01,Fav:etal:16} the stick-breaking decomposition of $\bm{\tp}$ is available meaning that we can decompose $\tp_1 = v_1$, and
\[
\tp_j = v_j\prod_{l=1}^{j-1}(1-v_l), \quad j \geq 2,
\]
for some sequence $\bm{v} = (v_j)_{j=1}^{\infty}$ for which its law is known. In this case, noting that $1-\sum_{l=1}^{j-1}\tp_l = \prod_{l=1}^{j-1}(1-v_l)$, we can update $\bm{\tp}$ via $\bm{v}$. We find
\begin{equation}\label{eq:v_post_sb}
  \pi(\bm{v}\mid\rest\,)
  \propto \bigg[\prod_{j=1}^{k_n}v_j^{n_j-1}(1-v_j)^{\sum_{l > j}n_l}
  \bigg]\times\pi(\bm{v}),
\end{equation}
For example, for the Pitman-Yor model we have that apriori, $v_j \ind  \mathsf{Be}(1-\sigma,\beta+j\sigma)$, $j \geq 1$, for some $\sigma \in [0,1)$ and $\beta > -\sigma$. Thus one would update elements of $\bm{v}$ independently of each other by sampling $v_j \sim \Be(n_j-\sigma,\sum_{l > j}n_l+\beta+j\sigma)$ for $j \leq k_n$ and $v_j \sim \Be(1-\sigma,\beta+j\sigma)$ for $j > k_n$.

\subsection*{Updating weights in order of appearance, $\bm{\tp}$, via $\bm{p}$ and $\bm{\alpha}$:}
As a notational device, here we keep \comillas{$\cdots$} to include all random terms other than $\bm{p}$ and $\bm{\alpha}$. From \eqref{eq:like} and \eqref{eq:prior} we find
\begin{equation}\label{eq:post_p_a}
\pi(\bm{p},\bm{\alpha}\mid \cdots) \propto \prod_{j=1}^{k_n}p_{\alpha_j}^{n_j} \times \prod_{j > k_n}  p_{\alpha_j}\bigg(1-\sum_{l=1}^{j-1}p_{\alpha_l}\bigg)^{-1}\Ind_{\clA} \times \pi(\bm{p}),
\end{equation}
The key idea to attain simple updating steps is to first update $\bm{\alpha}_{1:k_n}$ conditioning on $\bm{p}$ and $\bm{\alpha}_{k_n+1:\infty} = (\alpha_j)_{j > k_n}$ and later, update $\bm{p}$ and  $\bm{\alpha}_{k_n+1:\infty}$ as a block conditioning on $\bm{\alpha}_{1:k_n}$. Noting that for each $j > 1$, $\sum_{l=1}^{j-1}p_{\alpha_l} = 1-\sum_{l=j}^{\infty}p_{\alpha_l}$ we obtain
\begin{equation}\label{eq:post_a_kn}
\pi(\bm{\alpha}_{1:k_n}\mid \bm{p},\bm{\alpha}_{k_n+1:\infty},\cdots)\propto \prod_{j=1}^{k_n}p_{\alpha_j}^{n_j}\Ind_{\clA}.
\end{equation}
The event $\clA$ indicates that this is about sampling from a weighted permutation $\rho$ of the $k_n$ integers corresponding to the current values of $\bm{\alpha}_{1:k_n}$. In particular, we do not require to observe $\bm{\alpha}_{k_n+1:\infty}$ in order to sample from \eqref{eq:post_a_kn}, nor infinitely many weights. In fact, in order to update $\bm{\alpha}_{1:k_n}$, we can first sample $\rho$ from
\begin{equation}\label{eq:pi_rho}
\pi(\rho) = \frac{1}{Z}\prod_{j=1}^{k_n}w_{j,\rho(j)}, \quad \rho \in \mathcal{S}_{k_n}
\end{equation}
where $w_{j,l} = p_{\alpha_l}^{n_j}$, $\alpha_1,\ldots,\alpha_k$ are the current values of $\bm{\alpha}_{1:k_n}$, $Z$ is the normalizing constant, and $\mathcal{S}_{k_n}$ is the space of permutations of $\{1,\ldots,k_n\}$. Later we simply apply $\rho$ to the indexes of $\bm{\alpha}_{1:k_n}$ obtaining $(\alpha_{\rho(1)},\ldots,\alpha_{\rho(k_n)})$. Now, when $k_n$ is small, sampling from $\pi(\rho)$ can be done directly, by enumerating all possible permutations. Otherwise, we suggest to follow \cite{Zan:20} and adopt a Metropolis-Hastings scheme using a locally-balanced informed proposal (cf. Example 3 therein). For the simulation study, we chose the proposal distribution $Q(\rho^*\mid \rho) = \frac{1}{Z(\rho)}\sqrt{\pi(\rho^*)}\Ind_{\{\rho^*\in N(\rho)\}}$, where $Z(\rho) = \sum_{\rho^* \in N(\rho)}\sqrt{\pi(\rho^*)}$ and $N(\rho)$ is the neighbourhood of $\rho$ containing all permutations $\rho^* \in \mathcal{S}_{k_n}$ such that for some $i \neq j$: $\rho^*(i) = \rho(j)$, $\rho^*(j) = \rho(i)$ and $\rho^*(l) = \rho(l)$ for $l \not\in\{i,j\}$. Another alternative is to choose $2 \leq m < k_n$ such that it is possible to enumerate permutations on $\mathcal{S}_m$, randomly choose distinct indexes $\clJ = \{j_1,\ldots,j_m\} \subseteq \{1,\ldots,k_n\}$ and sample from
\[
\pi(\,(\rho(j_1),\ldots,\rho(j_m)) \mid (\rho(l))_{l \not\in \clJ}) \propto \prod_{i=1}^{m}w_{j_i,\rho(j_i)}, \quad \rho \in \mathcal{S}_{k_n}
\]
Later apply $\rho$ to the indexes of $\bm{\alpha}_{1:k_n}$, recompute the weights $w_{j,l}$ and repeat this step a few times. This is equivalent to sample repeatedly from
\[
\pi((\alpha_{j_1},\ldots,\alpha_{j_m}) \mid \bm{p}, (\alpha_l)_{l \not\in \clJ},\cdots) \propto \prod_{i=1}^{m}p_{\alpha_{j_i}}^{n_j}\Ind_{\clA},
\]
where $\clJ$ is chosen differently each time, thus attaining a Gibbs within Gibbs scheme.

As for the updating of $\bm{p}$ and $\bm{\alpha}_{k_n+1:\infty}$ as a block, we will first update elements in $\bm{p}$ integrating $\bm{\alpha}_{k_n+1:\infty}$ out and later we will update $\bm{\alpha}_{k_n+1:\infty}$ conditioning on elements of $\bm{p}$. We find
\[
\pi(\bm{p},\bm{\alpha}_{k_n+1:\infty}\mid \bm{\alpha}_{1:k_n},\cdots) \propto \prod_{j=1}^{k_n}p_{\alpha_j}^{n_j} \times \prod_{j > k_n} p_{\alpha_j}\bigg(1-\sum_{l=1}^{j-1}p_{\alpha_l}\bigg)^{-1}\Ind_{\clA} \times \pi(\bm{p}).
\]
Noting that $ \prod_{j=1}^{k_n}p_{\alpha_j}^{n_j} = \prod_{j=1}^{k_n}p_j^{r_j}$, where $r_j = \sum_{l=1}^{k_n}n_l\Ind_{\{\alpha_l = j\}}$, and summing over all possible values of $\bm{\alpha}_{k_n+1:\infty}$ we get
\begin{equation}\label{eq:post_p}
\pi(\,\bm{p}\mid \bm{\alpha}_{1:k_n},\cdots) \propto \prod_{j=1}^{\overline{\bm{\alpha}}_{k_n}}p_{j}^{r_j} \times \pi(\bm{p}).
\end{equation}
In particular, if the stick-breaking decomposition of $\bm{p}$ is available, i.e. we can write $p_j = v_j\prod_{l=1}^{j-1}(1-v_l)$, where the law of $\bm{v}$ is known. Then we can update $\bm{p}$ via sampling $\bm{v}$ from
\begin{equation}\label{eq:post_v}
\pi(\,\bm{v}\mid \bm{\alpha}_{1:k_n},\cdots) \propto \prod_{j=1}^{\overline{\bm{\alpha}}_{k_n}}v_{j}^{r_j}(1-v_j)^{\sum_{l > j}r_l} \times \pi(\bm{v}).
\end{equation}
For instance, if a priori $v_j\ind \Be(a_j,b_j)$, then we would update elements in $\bm{v}$ by independently sampling $v_j \ind \Be(r_j+a_j,\sum_{l>j}r_l+b_j)$ for $j \leq k_n$ and  $v_j\ind \Be(a_j,b_j)$ for $j > k_n$. Other examples can be found in the supplementary material of \cite{Deb:Gil:23}
%and \red{cite Markov draft}.
It only remains to specify how to update $\alpha_j$ for $j \geq k_n$. We have that
\begin{equation}\label{eq:post_a_new}
\pi(\bm{\alpha}_{k_n+1:\infty}\mid \bm{p},\bm{\alpha}_{1:k_n},\cdots) =  \prod_{j > k_n} p_{\alpha_j}\bigg(1-\sum_{l=1}^{j-1}p_{\alpha_l}\bigg)^{-1}\Ind_{\clA} = \prod_{j > k_n} \pi(\alpha_j\mid \bm{p},\bm{\alpha}_{1:j-1}).
\end{equation}
Hence $\alpha_{k_n+1},\alpha_{k_n+2},\ldots$ are sampled without replacement from $\{j\geq 1: j\not\in \bm{\alpha}_{1:k_n}\}$ with probabilities proportional to $\{p_j\geq 1: j\not\in \bm{\alpha}_{1:k_n}\}$, which can be achieved by sampling sequentially from \eqref{eq:size-biased_pick}, as a a priori.

Note that switches between indexes in $\bm{\alpha}_{1:k_n}$ and indexes in $\bm{\alpha}_{k_n+1:\infty}$ only occur when $k_n$ grows as a consequence of an update in $\bm{d}_{1:n}$. To facilitate the mixing one can include the \emph{acceleration step} described in Algorithm \ref{alg:alpha}, after updating $\bm{\alpha}_{1:k_n}$ from \eqref{eq:post_a_kn} and before updating $\bm{\alpha}_{k_n+1:\infty}$ from \eqref{eq:post_a_new}.

\begin{algorithm}[tb]
\caption{Acceleration step for $\bm{\alpha}$}\label{alg:alpha}
\begin{small}
For each $j \leq k_n$:
\begin{enumerate}
\item  Sample an auxiliary random variable $\alpha'_j$ from $\pi(\alpha_{k_n+1}\mid \bm{p},\bm{\alpha}_{1:k_n},\cdots)$.
\item Define
\[
\mathrm{q}_{\mathrm{preserve}} = p_{\alpha_j}^{n_j}p_{\alpha'_j}\bigg(1-\sum_{l=1}^{k_n}p_{\alpha_l}\bigg)^{-1}
\]
and
\[
\mathrm{q}_{\mathrm{switch}} =  p_{\alpha'_j}^{n_j}p_{\alpha_j}\bigg(1-\sum_{l=1}^{k_n}p_{\alpha_l}+p_{\alpha_j}-p_{\alpha'_j}\bigg)^{-1}
\]
\item With probability proportional to $\mathrm{q}_{\mathrm{switch}}$ set $\alpha_j = \alpha'_j$ and with probability proportional to $\mathrm{q}_{\mathrm{preserve}}$ do not modify $\alpha_j$.
\end{enumerate}
\end{small}
\end{algorithm}

Essentially, Algorithm \ref{alg:alpha} aims to draw samples from
\begin{equation}\label{eq:a_j_acc_step}
q(\alpha_j) = \pi(\alpha_j\mid (\alpha_l)_{l\leq k_n,l \neq j},\bm{p},\cdots) \propto p_{\alpha_j}^{n_j}\Ind_{\clA}, \quad \text{ for }j \leq k_n,
\end{equation}
allowing each $\alpha_j$ to take any integer value other than the ones in $(\alpha_l)_{l\leq k_n,l \neq j}$. To achieve this, we treat $\bm{\alpha}_{k_n+1:\infty}$ as a latent variable with finite dimensional distributions as in \eqref{eq:post_a_new}, and then update the pair $(\alpha_j,\alpha_{k_n+1})$ from
\[
\pi(\alpha_j,\alpha_{k_n+1}\mid \bm{p},(\alpha_l)_{l \not\in\{j,k_n+1\}},\cdots) \propto p_{\alpha_j}^{n_j}p_{\alpha_{k_n+1}}\bigg(1-\sum_{l=1}^{k_n}p_{\alpha_l}\bigg)^{-1},  \quad \text{ for }j \leq k_n.
\]

\begin{remark}\label{rem:1}
Notice that since $n$ data points cannot be generated by more than $n$ distinct components, i.e. $k_n \leq n$, at most the sampler will require to update $\tx_j$ and $\tp_j$ for $j \leq n$, and in practice much less. Namely, when updating component parameters and weights can just update $\bm{\tx}_{1:k_n}$, and $\bm{\tp_{1:k_n}}$ (possibly via $\bm{\alpha}_{1:k_n}$ and $\bm{p}_{1:\overline{\bm{\alpha}}_{k_n}}$). Later, when updating $d_i$, we can sample the required extra elements $\tx_j$ and $\tp_j$, for $j > k_n$ (possibly via $\alpha_j$ and sampling the necessary weights $p_l$ with $l > \overline{\bm{\alpha}}_{k_n}$) in view that a new of the mixture is discovered.
\end{remark}

\section{Distributional symmetries of weights in size-biased random order}\label{app:SBW}

\begin{lemma}\label{lem:G_k}
Let $\bm{\tp}$ be a sequence of weights that is invariant under size-biased permutation. For $k \geq 2$ define
\begin{equation}\label{eq:G_k}
G_k(\mathrm{d}p_1,\ldots,\mathrm{d}p_k) = \Prob(\tp_1 \in \ddr p_1,\ldots, \tp_k \in \ddr p_k)\prod_{j=1}^{k-1}\bigg(1-\sum_{l=1}^j p_l\bigg).
\end{equation}
Then $G_k$ is symmetric, i.e. it is invariant with respect to permutations of coordinates of $\R^k$. In particular, if the joint density $\pi(\bm{\tp}_{1:k}) = \pi(\tp_1,\ldots,\tp_k)$ of $\bm{\tp}_{1:k}$ exists,
\begin{equation}\label{eq:g_k}
g_k(\tp_1,\ldots,\tp_k) = \pi(\tp_1,\ldots,\tp_k)\prod_{j=1}^{k-1}\bigg(1-\sum_{l=1}^j \tp_l\bigg)
\end{equation}
is a symmetric function.
\end{lemma}

The proof of Lemma \eqref{lem:G_k} can be found in Theorem 4 and Corollary 5 by \cite{Pit:96}. For reader's convenience we recall it below

\begin{proof}
Let $\bm{p}$ be any rearrangement of $\bm{\tp}$, and define the exchangeable sequence $\bm{e} = (e_i)_{i=1}^{\infty}$ driven by $\sum_{j=1}^{\infty}p_j\delta_j$, this is $e_i\mid \bm{p} \iid \sum_{j=1}^{\infty} p_j \delta_j$. Also define sequentially $\tau_j = \inf\{i > \tau_{j-1}: e_i \not\in \{e_{\tau_1},\ldots,e_{\tau_{i-1}}\}\}$, with $\tau_1 = 1$. This way $\bm{\alpha} = (\alpha_j)_{j=1}^{\infty}$ given by $\alpha_j = e_{\tau_j}$ are the distinct values that $\bm{e}$ exhibits in order of appearance. We have that
\[
\Prob(\alpha_1 = l\mid \bm{p}) = \Prob(e_{\tau_1} = j\mid \bm{p}) = p_j
\]
and for $l \geq 1$,
\[
\Prob(\alpha_{l+1} = j\mid \bm{p},\alpha_1,\ldots,\alpha_l) = \Prob(e_{\tau_{l+1}} = j\mid \bm{p},e_{\tau_1},\ldots,e_{\tau_l}) = \frac{p_j}{1-\sum_{i=1}^l p_i}\Ind_{\{j \not \in \{\alpha_1,\ldots,\alpha_l\}\}},
\]
this is $\bm{\alpha}$ satisfies \eqref{eq:size-biased_pick}. It then follows from the definition of size-biased random permutations that $\bm{\tp} = (\tp_j)_{j=1}^{\infty}$ is identical in distribution to $\bm{\tilde{q}} = (\tilde{q}_j)_{j=1}^{\infty}$ where $\tilde{q}_j = p_{\alpha_j} = p_{e_{\tau_j}}$. Since we are concerned with a distributional property, we can assume without loss of generality that $\tp_j = p_{\alpha_j} = p_{e_{\tau_j}}$. Now, similarly as it occurs in \emph{iii} of Theorem \ref{theo:rep_sss} we have that
\begin{equation*}
\Prob(e_i \in \cdot\mid \bm{\tp},\alpha_{k_i+1},\bm{e}_{1:i}) = \sum_{j=1}^{k_i}\tp_j\delta_{\alpha_j} + \bigg(1-\sum_{j=1}^{k_i} \tp_j\bigg)\delta_{\alpha_{k_i+1}}
\end{equation*}
where $\alpha_1,\ldots,\alpha_{k_i}$ are the distinct values that $\bm{e}_{1:i}$ exhibits in order of appearance. Thus, the probability of the event $\mathcal{E}_k = \{e_1,\ldots,e_k \text{ are all distinct}\}$ is
\begin{equation}\label{eq:Ek}
\Prob(\Ind_{\mathcal{E}_k} = 1 \mid \bm{\tp}) = \Prob(\mathcal{E}_k \mid \bm{\tp}) = \prod_{j=1}^{k-1}\bigg(1-\sum_{l=1}^j \tp_l\bigg).
\end{equation}
Now, consider the random vector
\begin{equation*}\label{eq:tp_exch_Ek}
(\tp_1,\ldots,\tp_k)\Ind_{\mathcal{E}_k} = (p_{e_{\tau_1}},\ldots,p_{e_{\tau_k}})\Ind_{\mathcal{E}_k} = (p_{e_1},\ldots,p_{e_k})\Ind_{\mathcal{E}_k},
\end{equation*}
which is exchangeable because $e_i\mid \bm{p} \iid \sum_{j=1}^{\infty} p_j \delta_j$. It follows from \eqref{eq:Ek} that the measure, $G_k$, defined in \eqref{eq:G_k} and the distribution of $(\tp_1,\ldots,\tp_k)\Ind_{\mathcal{E}_k}$ are identical when restricted to $\R^k\setminus\{\bm{0}^k\}$, where $\bm{0}^k = (0,\ldots,0)$ is the origin in $\R^k$. Hence we must have that $G_k$ is symmetric. The second statement is straight forward when the joint density of $\bm{\tp_{1:k}}$ exists.
\end{proof}

For the rest of this section we will assume that the joint densities $\pi(\bm{\tp}_{1:k}) = \pi(\tp_1,\ldots,\tp_k)$ and $\pi(\bm{p}_{1:k}) = \pi(p_1,\ldots,p_k)$  exist. However, the following results can be generalized in the same way that \eqref{eq:G_k} generalizes \eqref{eq:g_k}.

\begin{proposition}\label{prop:p_k+1}
Let $\bm{\tp}$ be a sequence of weights that is invariant under size-biased permutations. Then, for every $k \geq 1$, the distribution of $\tp_{k+1}$ given $\bm{\tp}_{1:k}$ is invariant under permutations of $\bm{\tp}_{1:k}$. This is
\[
\pi(\tp_{k+1}\mid \bm{\tp}_{1:k}) = H_{k+1}(\tp_1,\ldots,\tp_k,\tp_{k+1})
\]
for some function $H_k$ that is invariant with respect to permutations of the first $k$ coordinates.
\end{proposition}

\begin{proof}
We know from Lemma \eqref{eq:G_k} that
\begin{align*}
g_{k+1}(\bm{\tp}_{1:k+1}) & = \pi(\bm{\tp}_{1:k+1})\prod_{j=1}^{k}\bigg(1-\sum_{l=1}^j \tp_l\bigg)\\
& = \pi(\tp_{k+1} \mid \bm{\tp}_{1:k})\pi(\bm{\tp}_{1:k})\bigg(1-\sum_{l=1}^k \tp_l\bigg)\prod_{j=1}^{k-1}\bigg(1-\sum_{l=1}^j \tp_l\bigg)\\
& =  \pi(\tp_{k+1} \mid \bm{\tp}_{1:k})\bigg(1-\sum_{l=1}^k \tp_l\bigg)g_k(\bm{\tp}_{1:k})
\end{align*}
where $g_{k+1}$ and $g_k$ are symmetric. Hence
\[
 \pi(\tp_{k+1} \mid \bm{\tp}_{1:k}) = \bigg(1-\sum_{l=1}^k \tp_l\bigg)^{-1} \frac{g_{k+1}(\bm{\tp}_{1:k+1})}{g_k(\bm{\tp}_{1:k})}
\]
is clearly invariant with respect to permutations of $\bm{\tp}_{1:k} = (\tp_1,\ldots,\tp_k)$.
\end{proof}

\begin{proposition}\label{prop:pa_k+1}
Let $\bm{p}$ be any sequence of weights and let $\bm{\alpha}$ be as in \eqref{eq:size-biased_pick}. Set $\bm{\alpha}_{1:k} = (\alpha_1,\ldots,\alpha_k)$, $\overline{\bm{\alpha}}_k = \max\{\bm{\alpha}_{1:k}\}$ and $\bm{p}_{1:\overline{\bm{\alpha}}_k} = (p_1,\ldots,p_{\overline{\bm{\alpha}}_k})$ for each $k \geq 1$. Then, for every $k \geq 1$, $\overline{\bm{\alpha}}_k \leq l$ and $L \geq \overline{\bm{\alpha}}_{k+1}$, the conditional density
\[
\pi(\bm{p}_{l+1:L},\alpha_{k+1}\mid \bm{\alpha}_{1:k},\bm{p}_{1:l})
\]
is invariant under permutations of $\bm{\alpha}_{1:k}$, where $\bm{p}_{l+1:L} = (p_{l+1},\ldots,p_{L})$, with the convention that $\bm{p}_{l+1:L} $ is empty if $l \geq L$.
\end{proposition}

\begin{proof}
First we note, from \eqref{eq:size-biased_pick}, that if $L \geq \overline{\bm{\alpha}}_{k}$,
\[
\pi(\bm{p}_{1:L},\bm{\alpha}_{1:k}) = \prod_{j = 1}^k  p_{\alpha_j}\bigg(1-\sum_{l=1}^{j-1}p_{\alpha_l}\bigg)^{-1} \Ind_{\clA}\,\pi(\bm{p}_{1:L}),
\]
where $\clA$ is the event that $\alpha_i \neq \alpha_j$ for $i \neq j$. Thus, for every $\overline{\bm{\alpha}}_{k} \leq l < L$
\[
\pi(\bm{p}_{l+1:L}\mid \bm{\alpha}_{1:k}, \bm{p}_{1:l}) \propto \pi(\bm{p}_{1:L}) \propto \pi(\bm{p}_{l+1:L} \mid \bm{p}_{1:l}).
\]
This means that for $l \geq \overline{\bm{\alpha}}_{k}$, $\bm{p}_{l+1:L}$ is conditionally independent of $\bm{\alpha}_{1:k}$ given $\bm{p}_{1:l}$. Now, from \eqref{eq:size-biased_pick} we also have
\[
\pi(\alpha_{k+1}\mid \bm{\alpha}_{1:k}, \bm{p}) = p_{\alpha_{k+1}}\bigg(1-\sum_{l=1}^{k}p_{\alpha_l}\bigg)^{-1}\Ind_{\clA}.
\]
Noting that, in terms of the weights, $\pi(\alpha_{k+1}\mid \bm{\alpha}_{1:k}, \bm{p})$ at most depends on $p_{\overline{\bm{\alpha}}_{k+1}}$, we get
\[
\pi(\bm{p}_{l+1:L},\alpha_{k+1}\mid \bm{\alpha}_{1:k},\bm{p}_{1:l}) = p_{\alpha_{k+1}}\bigg(1-\sum_{l=1}^{k}p_{\alpha_l}\bigg)^{-1}\Ind_{\clA}\,\pi(\bm{p}_{l+1:L} \mid \bm{p}_{1:l})
\]
which is clearly symmetric with respect to $\bm{\alpha}_{1:k}$.
\end{proof}

\begin{remark}\label{rem:pa_k+1}
Proposition \ref{prop:pa_k+1} proves that conditioning on the first $l$ weights, $\bm{p}_{1:l}$ as well as $\bm{\alpha}_{1:k}$ as in \eqref{eq:size-biased_pick}, with $p_{\alpha_j} \leq p_l$ for every $j \leq k$, the distribution subsequent weights, $\bm{p}_{l+1:L}$, and $\alpha_{k+1}$, depends on $\bm{\alpha}_{1:k}$ symmetrically. This is 
\begin{equation}\label{eq:pa_k+1}
\pi(\bm{p}_{l+1:L},\alpha_{k+1}\mid \bm{\alpha}_{1:k},\bm{p}_{1:l}) = \pi(\bm{p}_{l+1:L},\alpha_{k+1}\mid \bm{\alpha}^*_{1:k},\bm{p}_{1:l})
\end{equation}
for every permutation of  $\bm{\alpha}^*_{1:k} = (\alpha^*_1,\ldots,\alpha^*_k)$ of  $\bm{\alpha}_{1:k} = (\alpha_1,\ldots,\alpha_k)$. To explain the implications of this result for the OAS, first note that it follows from \eqref{eq:post_p_a} that
\begin{equation}\label{eq:pa_k+1_2}
\pi(\bm{p}_{l+1:L},\alpha_{k+1}\mid \bm{\alpha}_{1:k},\bm{p}_{1:l},\rest) = \pi(\bm{p}_{l+1:L},\alpha_{k+1}\mid \bm{\alpha}_{1:k},\bm{p}_{1:l}) 
\end{equation}
for $l \geq \overline{\bm{\alpha}}_{k}$, $k+1 = k_n$ and $n_{k+1} = 1$. Here, as a notational device, \comillas{$\rest$} includes all random terms in the OAS excluding $\bm{p}$ and $\bm{\alpha}$. Equations \eqref{eq:pa_k+1} and \eqref{eq:pa_k+1_2} yield that in order to update a newly discovered weight, $p^*_{k+1} = p_{\alpha^*_{k+1}}$ as 1.(ii) of Algorithm \ref{alg:EOAS}, it is enough to sample the necessary extra weights and the new index, $(\bm{p}_{l+1:L},\alpha^*_{k+1})$, from
\[
\pi(\bm{p}_{l+1:L},\alpha^*_{k+1}\mid \bm{\alpha}^*_{1:k},\bm{p}_{1:l},\rest) = \pi(\bm{p}_{l+1:L},\alpha^*_{k+1}\mid \bm{\alpha}^*_{1:k},\bm{p}_{1:l}),
\]
identically as one samples $(\bm{p}_{l+1:L},\alpha_{k+1})$, from $\pi(\bm{p}_{l+1:L},\alpha_{k+1}\mid \bm{\alpha}_{1:k},\bm{p}_{1:l})$ corresponding to the prior law $\pi(\bm{p},\bm{\alpha}) = \pi(\bm{\alpha}\mid \bm{p})\pi(\bm{p})$. Furthermore, notice that sampling from this distribution is easy because the normalization constant is known.
\end{remark}

\end{document}